\documentclass[a4paper,english,11pt]{article}
\usepackage[T1]{fontenc}
\usepackage[latin9]{inputenc}

\usepackage[top=1in, bottom=1in, left=1in, right=1in]{geometry}

\usepackage{verbatim}
\usepackage{float}
\usepackage{amssymb,amsfonts,amsmath}
\usepackage{enumerate}
\usepackage{algorithm}
\usepackage[noend]{algpseudocode}

\usepackage{color}
\usepackage{graphicx}
\usepackage{hyperref}
\usepackage{amsthm}

\usepackage{caption}
\usepackage{subfigure}

\floatstyle{ruled}
\newfloat{Figure}{tbp}{loa}
\floatname{Figure}{Figure}

\newtheorem{proposition}{Proposition}

\DeclareGraphicsExtensions{.eps}

\newcommand{\eps}{\varepsilon}
\newcommand{\tO}{{\tilde{O}}}
\newcommand{\MM}{\mathcal{M}}

\newtheorem{lemma}{Lemma}
\newtheorem{corollary}{Corollary}
\newtheorem{theorem}{Theorem}
\newtheorem{definition}{Definition}

\newtheorem{claim}{Claim}

\newcommand{\minp}{\star}

\newcommand{\Ex}{{\mathbb{E}}}

\newcommand{\map}[2]{\textit{map}_{#1}(#2)}

\newcommand{\poly}{\textrm{poly}}

\def\DEBUG{true}

\ifdefined\DEBUG

  \def\rem#1{{\marginpar{\raggedright\scriptsize #1}}}
  \newcommand{\fabr}[1]{\rem{\textcolor{red}{$\bullet$ #1}}}
  \newcommand{\karr}[1]{\rem{\textcolor{blue}{$\bullet$ #1}}}
  \newcommand{\barnr}[1]{\rem{\textcolor{green}{$\bullet$ #1}}}
  \newcommand{\virr}[1]{\rem{\textcolor{cyan}{$\bullet$ #1}}}
\else

  \newcommand{\fabr}[1]{}
  \newcommand{\karr}[1]{}
  \newcommand{\barnr}[1]{}
  \newcommand{\virr}[1]{}

\fi

\author{Karl Bringmann\thanks{Max Planck Institute for Informatics, Saarland Informatics Campus, Saarbr\"ucken, Germany, {\tt kbringma@mpi-inf.mpg.de}.}, Fabrizio Grandoni\thanks{IDSIA, University of Lugano, {\tt fabrizio@idsia.ch}. This work was partially supported by the ERC StG project NEWNET no.~279352 and the SNSF project APPROXNET no.~200021\_159697/1.}, Barna Saha\thanks{University of Massachusetts Amherst, College of Information and Computer Science, Amherst, MA. {\tt barna@cs.umass.edu.} This work is partially supported by NSF Grant CCF-1464310, NSF CAREER CCF-1652303, a Yahoo ACE Award and a Google Faculty Research Award.} , Virginia Vassilevska Williams\thanks{Massachusetts Institute of Technology, CSAIL, {\tt virgi@mit.edu}. Partially supported by NSF Grants CCF-1417238, CCF-1528078 and CCF-1514339, and BSF Grant BSF:2012338. }} 

\title{Truly Sub-cubic Algorithms for\\ Language Edit Distance and RNA Folding\\ via Fast Bounded-Difference Min-Plus Product\footnote{ This work was done in part while the authors were visiting the Simons Institute for the Theory of Computing. The fourth author was at Stanford University at the time.}}

\begin{document}

\maketitle

\thispagestyle{empty}
\setcounter{page}{0}

\begin{abstract}
\noindent It is a major open problem whether the $(\min,+)$-product of two $n\times n$ matrices has a truly sub-cubic (i.e. $O(n^{3-\eps})$ for $\eps>0$) time algorithm, in particular since it is equivalent to the famous 
All-Pairs-Shortest-Paths problem (APSP) in $n$-vertex graphs.
Some restrictions of the $(\min,+)$-product to special types of matrices are known to admit truly sub-cubic algorithms, each giving rise to a special case of APSP that can be solved faster. In this paper we consider a new, different and powerful restriction in which all matrix entries are integers and one matrix can be arbitrary, as long as the other matrix has ``bounded differences'' in either its columns or rows, i.e. any two consecutive entries differ by only a small amount. We obtain the first truly sub-cubic algorithm for this bounded-difference $(\min,+)$-product (answering an open problem of Chan and Lewenstein). 

Our new algorithm, combined with a strengthening of an approach of L.~Valiant for solving context-free grammar parsing with matrix multiplication, yields the first truly sub-cubic algorithms for the following problems: Language Edit Distance (a major problem in the parsing community), RNA-folding (a major problem in bioinformatics) and Optimum Stack Generation (answering an open problem of Tarjan).
\end{abstract}

\newpage

\section{Introduction}

The $(\min,+)$-product (also called \emph{min-plus} or {\em distance} product) of two integer matrices $A$ and $B$ is the matrix $C=A\minp B$ such that $C_{i,j} = \min_k \{A_{i,k}+B_{k,j}\}$.\footnote{By $M_{i,j}$ we will denote the entry in row $i$ and column $j$ of matrix $M$.} Computing a $(\min ,+)$-product is a basic primitive used in solving many other problems. For instance, Fischer and Meyer~\cite{FischerM71} showed that the $(\min,+)$-product of two $n\times n$ matrices has essentially the same time complexity as that of the All Pairs Shortest Paths problem (APSP) in $n$-node graphs, one of the most basic problems in graph algorithms. APSP itself has a multitude of applications, from computing graph parameters such as the diameter, radius and girth, to computing replacement paths and distance sensitivity oracles (e.g.~\cite{BernsteinK09,WilliamsW10,GrandoniW12}) and vertex centrality measures (e.g.~\cite{brandes,AbboudGW15}).

While the $(\min, +)$-product of two $n\times n$ matrices has a trivial $O(n^3)$ time algorithm, it is a major open problem whether there is a \emph{truly sub-cubic} algorithm for this problem, i.e. an $O(n^{3-\eps})$ time algorithm for some constant $\eps>0$. Following a multitude of polylogarithmic improvements over $n^3$ (e.g. \cite{fredman1976new,takaoka,Chan10}), a relatively recent breakthrough of Williams~\cite{rwill-apsp} gave an $O(n^3/c^{\sqrt{\log n}})$ time algorithm for a constant $c >1$. 
Note that despite this striking improvement, the running time is still not truly sub-cubic.

For restricted types of matrices, truly sub-cubic algorithms are known. The probably most relevant examples are: 
\begin{itemize}\itemsep0pt
\item[(1)] when all matrix entries are integers bounded in absolute value by $M$, then the problem can be solved in $\tilde{O}(Mn^\omega)$ time\footnote{The $\tilde O$-notation hides logarithmic factors, i.e., $\tilde O(T) = O(T \cdot \textup{polylog}(T))$.}~\cite{Alon:1997}, where $\omega<2.373$ is the matrix multiplication exponent~\cite{v12,Gall14a}; 
\item[(2)] when each row of matrix $A$ has at most $D$ distinct values, then the $(\min,+)$-product of $A$ with an arbitrary matrix $B$ can be computed in time $\tilde{O}(Dn^{(3+\omega)/2})$~\cite{Chan10,Yuster09}.\footnote{The same holds if $A$ is arbitrary and $B$ has at most $D$ distinct values per column.}
\end{itemize}
Among other applications, these restricted $(\min,+)$-products yield faster algorithms for special cases of APSP. E.g., the distance product (1) is used to compute APSP in both undirected~\cite{Seidel95,ShoshanZ99} and directed~\cite{zwickbridge} graphs with bounded edge weights, while the distance product (2) is used to compute APSP in graphs in which each vertex has a bounded number of distinct edge weights on its incident edges~\cite{Chan10,Yuster09}.

\subsection{Our Result}

In this paper we significantly extend the family of matrices for which a $(\min, +)$-product can be computed in truly sub-cubic time to include the following class.  
\begin{definition}
A matrix $X$ with integer entries is a \emph{$W$-bounded-difference} ($W$-BD) matrix if for every row $i$ and every column $j$, the following holds:
$$
|X_{i,j}-X_{i,j+1}|\leq W \quad\text{and}\quad |X_{i,j}-X_{i+1,j}|\leq W.
$$
When $W=O(1)$, we will refer to $X$ as a bounded-difference (BD) matrix.  
\end{definition}

In this paper we present the first truly sub-cubic algorithm for $(\min, +)$-product of BD matrices, answering a question of Chan and Lewenstein~\cite{cl:15}.
%\fabr{Given the discussion below, maybe we should directly give the thr for the more general case} 
\begin{theorem}\label{thr:mainProduct}
There is an $O(n^{2.8244})$ time randomized algorithm and an $O(n^{2.8603})$ time deterministic algorithm that computes the $(\min,+)$-product of any two $n\times n$ BD matrices.
\end{theorem}

Indeed, our algorithm produces a truly sub-cubic running time for $W$-BD matrices for nonconstant values of $W$ as well, as long as $W=O(n^{3-\omega-\eps})$ for some constant $\eps>0$. In fact, we are able to prove an even more general result: suppose that matrix $A$ only has bounded differences in its rows or its columns (and not necessarily both). Then, $A$ can be $(\min,+)$-multiplied by an arbitrary matrix $B$ in truly sub-cubic time:

\begin{theorem} \label{thm:general}
Let $A,B$ be integer matrices, where $B$ is arbitrary and we assume either of the following:
$$
\text{(i) $\forall i,j\in [n]$, $|A_{i,j}-A_{i+1,j}|\leq W$  \quad or \quad (ii) $\forall i,j\in [n]$, $|A_{i,j}-A_{i,j+1}|\leq W$.}
$$ 
If $W \le O(n^{3-\omega-\eps})$ for any $\eps>0$, then $A\minp B$ can be computed in randomized time $O(n^{3 - \Omega(\eps)})$. If $W=O(1)$, then $A\minp B$ can be computed in randomized time $O(n^{2.9217})$.
\end{theorem}

The main obstacle towards achieving a truly sub-cubic algorithm for the $(\min, +)$-product in general is the presence of entries of large absolute value. 
%Let us compare our result with (1) and (2) from this point of view.
In order to compare our result with (1) and (2) from that point of view, assume for a moment that $\omega=2$ (as conjectured by many). Then (1) can perform 
a $(\min, +)$-product in truly sub-cubic time if \emph{both} $A$ and $B$ have entries of absolute value at most $M=O(n^{1-\eps})$ for some constant $\eps>0$, 
while (2), without any other assumptions on $A$ and $B$, achieves the same if \emph{at least} one of $A$ 
and $B$ has entries of absolute value at most $M=O(n^{1/2-\eps})$. We can do the same when \emph{at least one} of $A$ and $B$ has entries of absolute value at most $M=O(n^{1-\eps})$.

\subsection{Our Approach}
Our approach has three phases. 
\paragraph{Phase 1: additive approximation \boldmath$\tilde{C}$ of the product \boldmath$C=A\minp B$}
%Consider the product $C=A\minp B$. 
For BD matrices it is quite easy to obtain an additive overestimate $\tilde{C}$ of $C$:
Let us subdivide $A$ and $B$ into square blocks of size $\Delta \times \Delta$, for some small polynomial value $\Delta = n^\delta$. Thus the overall product reduces to the multiplication of $O((n/\Delta)^3)$ pairs of blocks $(A',B')$. By the bounded-difference property, it is sufficient to compute $A'_{i,k}+B'_{k,j}$ for \emph{some} triple of indices $(i,k,j)$ in order to obtain an overestimate of \emph{all} the entries in $A' \minp B'$ within an additive error of  $O(\Delta W)$. This way in truly sub-cubic time we can  compute an  additive $O(\Delta W)$ overestimate $\tilde{C}$ of $C$.

%\fabr{Where do we exactly need this. Thr 2? If yes, this can be moved to appendix maybe}
\emph{Remark:} It would seem that Phase 1 requires that the matrices are BD, and one would not be able to use the same approach to attack the $(\min,+)$-product of general matrices. We note that this is NOT the case: Phase 1 can be performed for {\em arbitrary} integer matrices $A$ and $B$ as well, provided one has an algorithm that, given a very good approximation $\tilde C$, can compute the correct product $C$; this is exactly what the remaining phases do. To show this, we use a scaling approach \`a la Seidel~\cite{Seidel95}. Assume that the entries of $A$ and $B$ are nonnegative integers bounded by $M$, and obtain $A'$ and $B'$ by setting $A'_{i,j}=\lceil A_{i,j}/2\rceil$ and $B'_{i,j}=\lceil B_{i,j}/2\rceil$. Recursively compute $A'\minp B'$, where the depth of the recursion is $\log M$ and the base case is when the entries of $A$ and $B$ are bounded by a constant, in which case $A'\minp B'$ can be computed in $O(n^\omega)$ time. Then we can set $\tilde{C}_{i,j}=2C_{i,j}$ for all $i,j$. This gives an overestimate that errs by at most an additive $2$ in each entry. Thus, if all remaining phases  (which compute the correct product $C$ from the approximation $\tilde C$) could be made to work for arbitrary matrices, then Phase 1 would also work.

\paragraph{Phase 2: Correcting \boldmath$\tilde{C}$ up to a few bad triples}
The heart of our approach comes at this point. We perform a (non-trivial) perturbation of $A$ and $B$, and then set to $\infty$ the entries of absolute value larger than $c\cdot \Delta W$ for an appropriate constant $c$. The perturbation consists of adding the same vector $V^r_A$ (resp., $V^r_B$) to each column of $A$ (resp., row of $B$). Here $V^r_A$ and $V^r_B$ are random vectors derived from the estimate $\tilde{C}$.
Let $A^r$ and $B^r$ be the resulting matrices. Using result (1) from \cite{Alon:1997}, we can compute $C^r = A^r \minp B^r$ in truly sub-cubic time $O(\Delta W n^\omega)$ for sufficiently small $W$ and $\Delta$. The perturbation is such that it is possible to derive from $(C^r)_{i,j}$ the corresponding value $(A\minp B)_{i,j}=A_{i,k}+B_{k,j}$ \emph{unless} one of the entries $A^r_{i,k}$ or $B^r_{k,j}$ was \emph{rounded to $\infty$}.   

The crux of our analysis is to show that after $\rho$ rounds of perturbations and associated bounded entry $(\min,+)$-products, there are at most $\tilde O(n^{3}/\rho^{1/3})$ triples $(i,k,j)$ for which (a) $|A_{i,k}+B_{k,j}-\tilde{C}_{i,j}|\leq O( \Delta W)$ (i.e. $k$ is a potential witness for $C_{i,j}$) and (b) none of the perturbations had both $A^r_{i,k}$ and $B^r_{k,j}$ finite. 

Interestingly, our proof of correctness of Phase 2 relies on an extremal graph theoretical lemma that bounds from below the number of $4$-cycles in sufficiently dense bipartite graphs.

In a sense Phase 1 and 2 only leave $\tilde O(n^{3}/\rho^{1/3})$ work to be done: if we knew the ``bad'' triples that are not covered by the perturbation steps, we could simply iterate over them in a brute-force way, fixing $\tilde{C}$ to the correct product $C$.
Since Phases 1 and 2 do not use the fact that $A$ and $B$ are BD, if we could find the bad triples efficiently we
would obtain a truly sub-cubic algorithm for the $(\min,+)$-product!

\paragraph{Phase 3: Finding and fixing the bad triples}
To fix the bad triples, one could try to keep track of the triples covered in each perturbation iteration. For arbitrary matrices $A$ and $B$ this would not give a truly sub-cubic algorithm as the number of triples is already $n^3$. For BD matrices, however, we do not need to keep track of all triples, but it suffices to consider the triples formed by the upper-most left-most entries of the blocks from Phase 1, since these entries are good additive approximations of all block entries. The number of these block representative triples is only $O((n/\Delta)^3)$ where $\Delta$ is the block size (from Phase 1). Thus, instead of spending at least $n^3$ time, we obtain an algorithm spending $O(\rho\cdot (n/\Delta)^3)$ time, where $\rho$ is the number of perturbation rounds (from Phase 2).
%\fabr{what is $\rho$? $n^{3d}$? Also, $(n/\Delta)^3$ is the total number of representatives, the bad ones should be $(n/\Delta)^3\cdot n^{-d}$ or so} \virr{- it doesn't matter as long as they are sub-cubic. I'm ignoring the rest of savings}
After finding the bad block representative triples, we can iterate over their blocks in a brute-force manner to fix $\tilde{C}$ and compute $C$.
Since each triple in the blocks of a bad block representative triple must also be bad, the total number of triples considered by the brute-force procedure is $\tilde O(n^{3}/\rho^{1/3})$ as this is the total number of bad triples. 

We reiterate that this is the {\em only} phase of the algorithm that does not work for arbitrary matrices $A$ and $B$.

\subsection{Applications}

The notion of BD matrices is quite natural and has several applications. Indeed, our original motivation for studying the $(\min,+)$-product of such matrices came from a natural scored version of the classical Context-Free Grammar (CFG) parsing problem. It turns out that a fast algorithm for a bounded-difference version of scored parsing implies the first truly sub-cubic algorithms for some well-studied problems such as Language Edit Distance, RNA-Folding and Optimum Stack Generation.  

Recall that in the \emph{parsing} problem we are given a CFG $G$ and a string $\sigma=\sigma_1\ldots\sigma_n$ of $n$ terminals. Our goal is to determine whether $\sigma$ belongs to the language $L$ generated by $G$. For ease of presentation and since this covers most applications, we will assume unless differently stated that the size of the grammar is $|G|=O(1)$, and we will not explicitly mention the dependency of running times on the grammar size.\footnote{Our approach also works when $|G|$ is a sufficiently small polynomial.} We will also assume that $G$ is given in Chomsky Normal Form (CNF).\footnote{ Note that it is well-known that any context free grammar can be transformed into an equivalent CNF grammar.}
In a breakthrough result Valiant~\cite{v97} proved a reduction from parsing to Boolean matrix multiplication: the parsing problem can be solved in $O(n^\omega)$ time.

One can naturally define a scored generalization of the parsing problem (see, e.g., \cite{ap72}). Here each production rule $p$ in $G$ has an associated non-negative integer score (or cost) $s(p)$. The goal is to find a sequence of production rules of minimum total score that generates a given string $\sigma$. It is relatively easy to adapt Valiant's parser to this scored parsing problem, the main difference being that Boolean matrix multiplications are replaced by $(\min,+)$-products. It follows that scored parsing can be solved up to logarithmic factors in the time needed to perform one $(\min,+)$-product (see also \cite{s:15}). In particular, applying Williams' algorithm for the $(\min,+)$-product \cite{rwill-apsp}, one can solve scored parsing in $O(n^3/2^{\Theta(\sqrt{\log{n}})})$ time, which is the current best running time for this problem.

For a nonterminal $X$ let $s(X,\sigma)$ be the minimum total score needed to generate string $\sigma$ from $X$ (where the grammar $G$ is assumed to be clear from the context). Let us define a bounded-difference notion for CFGs. Intuitively, we require that adding or deleting a terminal at one endpoint of a string does not change the corresponding score by much.
\begin{definition} \label{def:grammarBD}
A CFG $G$ is a $W$-bounded-difference ($W$-BD) grammar if, for any non-terminal~$X$, terminal $x$, and non-empty string of terminals $\sigma$, the following holds:
$$
|s(X,\sigma)-s(X,\sigma x)|\leq W \quad\text{and}\quad |s(X,\sigma)-s(X,x\sigma)|\leq W.
$$
When $W=O(1)$, we will refer to $G$ as a bounded-difference (BD) grammar.  
\end{definition}   

Via a simple but very careful analysis of the scored version of Valiant's parser, we are able to show that the scored parsing problem on BD grammars can be reduced to the $(\min,+)$-product of BD matrices (see Section \ref{sec:scoredParsing}).
\begin{theorem}\label{thr:SP}
Let $O(n^\alpha)$ be the time needed to perform one $(\min,+)$-product of two $n\times n$ BD matrices. Then the scored parsing problem on BD grammars in CNF can be solved in time $\tO(n^\alpha)$.
\end{theorem}
\begin{corollary}\label{cor:SP}
The scored parsing problem on BD grammars in CNF can be solved in randomized time $\tilde O(n^{2.8244})$ and deterministic time $\tilde O(n^{2.8603})$.
\end{corollary}

BD grammars appear naturally in relevant applications. Consider for example the well-studied Language Edit Distance problem (LED) \cite{ap72,m95,l97,s:14,s:15,abw:15,rn14}. Here we are given a CFG $G$ and a string $\sigma$ of terminals. We are allowed to \emph{edit} $\sigma$ by \emph{inserting}, \emph{deleting} and \emph{substituting} terminals. Our goal is to find a sequence of such edit operations of minimum length so that the resulting string~$\sigma'$ belongs to the language $L$ generated by $G$.\footnote{In some variants of the problem each edit operation has some integer cost upper bounded by a constant. Our approach clearly works also in that case.}
As already observed by Aho and Peterson in 1972 
\cite{ap72}, LED can be reduced to scored parsing. Indeed, it is sufficient to assign score zero to the production rules of the input grammar, and then augment the grammar with production rules of score $0$ and $1$ that model edit operations. We show that, by performing the above steps carefully, the resulting scored grammar is BD, leading to a truly sub-cubic algorithm for LED via Corollary \ref{cor:SP} (see Section \ref{sec:LED}). We remark that finding a truly sub-cubic algorithm for LED was wide open even for very restricted cases. For example, consider \emph{Dyck LED}, where the underlying CFG represents well-balanced strings of parentheses. %\fabr{On wikipedia Dyck refers to only one type of parenthesis. Here we mean the multi-parenthesis case. Is this ok?} \vir{I think so..}
Developing fast algorithms for Dyck LED and understanding the parsing problem for the parenthesis grammar has recently received considerable attention  \cite{ao:16,s:14,kls:mfcs11,cckm:focs10,mmn:stoc10,prr:rand03}. Even for such restricted grammars no truly sub-cubic exact algorithm was known prior to this work. 

Another relevant application is related to \emph{RNA-folding}, a central problem in bioinformatics defined by Nussinov and Jacobson in 1980 \cite{nj:80}. They proposed the following optimization problem, and a simple $O(n^3)$ dynamic programming solution to obtain the optimal folding. Let $\Sigma$ be a set of letters and let $\Sigma'=\{c' \mid c \in \Sigma\}$ be the set of ``matching'' letters, such that for every letter $c \in \Sigma$ the pair $c, c'$ matches. Given a sequence of $n$ letters over $\Sigma \cup \Sigma'$, the RNA-folding problem asks for the maximum number of non-crossing pairs $\{i,j\}$ such that the $i$th and $j$th letter in the sequence match. In particular, if letters in positions $i$ and $j$ are paired and if letters in positions $k$ and $l$ are paired, and $i < k$ then either they are nested, i.e., $i < k < l< j$ or they are non-intersecting, i.e., $ i < j < k < l$. (In nature, there are $4$ types of nucleotides in an RNA molecule, with matching pairs $A,U$ and $C,G$, i.e., $|\Sigma| = 2$.)
We can rephrase RNA-folding as follows. We are given the CFG with productions $S\rightarrow SS \mid \eps$ and $S \to \sigma S \sigma' \mid \sigma' S \sigma$ for any $\sigma \in \Sigma$ with matching $\sigma' \in \Sigma'$. The goal is to find the minimum number of {\em insertions} and {\em deletions} of symbols on a given string $\sigma$ that will generate a string $\sigma'$ consistent with the above grammar. This is essentially a variant of LED where only insertions and deletions (and no substitutions) are allowed. 
%Nussinov and Jacobson proposed a simple $O(n^3)$ time algorithm to solve RNA-folding. 
Despite  considerable efforts (e.g. \cite{vgf:13,akutsu1999,zakov2010,nj:80}), no truly sub-cubic algorithm for RNA-folding was known prior to our work. By essentially the same argument as for LED, it is easy to obtain a BD scored grammar modeling RNA-folding. Thus we immediately obtain a truly sub-cubic algorithm to solve this problem via Corollary \ref{cor:SP}.

As a final application, consider the Optimum Stack Generation problem (OSG) described by Tarjan in~\cite{tarjan2005problems}. Here, we are given a finite alphabet $\Sigma$, a stack $S$, and a string $\sigma \in \Sigma^*$. We would like to print $\sigma$ by a minimum length sequence of three stack operations: $\textit{push}()$, $\textit{emit}$ (i.e., print the top character in the  stack), and $\textit{pop}$, ending in an empty stack.
For example, the string $BCCAB$ can be printed via the following sequence of operations: $\textit{push}(B)$, $\textit{emit}(B)$, $\textit{push}(C)$, $\textit{emit}(C)$, $\textit{emit}(C)$, $\textit{pop}(C)$, $\textit{push}(A)$, $\textit{emit}(A)$, $\textit{pop}(A)$, $\textit{emit}(B)$, $\textit{pop}(B)$. While there is a simple $O(n^3)$ time algorithm for OSG, Tarjan suspected this could be improved. In Section \ref{sec:osg}, we show that OSG can be reduced to scored parsing on BD grammars. This leads to the first truly sub-cubic algorithm for OSG.

Let us summarize the mentioned applications of our approach. 
\begin{theorem}\label{thr:applications}
LED, RNA-folding, and OSG can be solved in randomized time $\tilde{O}(n^{2.8244})$ and deterministic time $\tilde{O}(n^{2.8603})$ (on constant-size grammars or alphabet, respectively).
\end{theorem}

Moreover, our techniques also lead to a truly subquadratic
algorithm for bounded monotone \emph{$(\min,+)$-convolution}. A subquadratic
algorithm was already and very recently achieved in a breakthrough
result by Chan and Lewenstein \cite{cl:15}, however with very different
techniques.
For two sequences $a=(a_1,\ldots,a_n)$ and $b=(b_1,\ldots,b_n)$ the
$(\min,+)$-convolution of $a$ and $b$ is the vector $c=(c_1,\ldots,c_n)$ with $c_k = \min_i \{a_i +
b_{k-i}\}$. Assume $n = m^2$.
A standard reduction from $(\min,+)$-convolution to the $(\min,+)$-matrix
product constructs the $m \times m$ matrices $A^r$ with $A^r_{i,k} =
a_{r m + i + k}$
(for $1 \le r \le m$)
and $B$ with $B_{k,j} = b_{j m - k}$.
Then from the products $A^r \minp B$ we can infer the
$(\min,+)$-convolution of $a$ and $b$ in time $O(n^{3/2})$. Note that if
$a$ has bounded differences, then the matrices $A^r$  have bounded differences along the rows, while if $b$ has bounded differences, then $B$ has bounded differences along the columns.
Theorem~\ref{thm:general} now allows us to compute the $m$ $(\min,+)$-products in time
$O(m \cdot m^{2.9217}) = O(n^{1.961})$, obtaining a subquadratic algorithm for BD $(\min,+)$-convolution. Previously, Chan and Lewenstein~\cite{cl:15} observed that computing the $(\min,+)$-convolution over bounded monotone sequences is equivalent to computing it over bounded-difference sequences, and presented an $O(n^{1.859})$ time algorithm for this case. Thus, our algorithm is not faster, but it works in the more general setting of $(\min,+)$-matrix multiplication.

We envision other applications of our BD $(\min,+)$-product algorithm to come in the future.

\subsection{Related Work}

\paragraph{Language Edit Distance} 

LED is among the most fundamental and best studied problems related to strings and grammars \cite{ap72,m95,l97,s:14,s:15,abw:15,rn14}. It generalizes two basic problems in computer science: parsing and string edit distance computation. In 1972, Aho and Peterson presented a dynamic programming algorithm for LED that runs in time $O(|G|^2n^3)$~\cite{ap72}, which was improved to $O(|G|n^3)$ by Myers in 1985 \cite{m95}. These algorithms are based on the popular CYK parsing algorithm \cite{aho:book} with the observation that LED can be reduced to a scored parsing problem~\cite{ap72}. This implies the previous best running time of $O(n^3/2^{{\Theta}(\sqrt{\log{n}})})$.
In a recent paper~\cite{s:15}, Saha showed that LED can be solved in $O(\frac{n^\omega}{\text{poly}(\epsilon)})$ time if we allow to approximate the exact edit distance by a $(1+\epsilon)$-factor. Due to known conditional lower bound results for parsing~\cite{l97,abw:15}, LED cannot be approximated within any multiplicative factor in time $o(n^\omega)$ (unless cliques can be found faster). Interestingly, if we only ask for insertions as edit operations, Sahe also showed that a truly sub-cubic exact algorithm is unlikely due to a reduction from APSP~\cite{s:15,WilliamsW10}. In contrast, here we show that with insertions and deletions (and possibly substitutions) as edit operations, LED is solvable in truly sub-cubic time. LED provides a very generic framework for modeling problems with many applications (e.g. \cite{kssy:13,Johnson:2010,speech,Moore:2002,video3,p82,ps:101}). A fast exact algorithm for it is likely to have tangible impact.

\paragraph{RNA-Folding}
Computational approaches to finding the secondary structure of RNA molecules are used extensively in bioinformatics applications. 
Since the seminal work of Nussinov and Jacobson \cite{nj:80}, a multitude of sophisticated RNA-folding algorithms with complex objectives and softwares have been developed,\footnote{see \url{https://en.wikipedia.org/wiki/List_of_RNA_structure_prediction_software}} but the basic dynamic programming algorithm of Nussinov and Jacobson remains at the heart of all of these. Despite much effort, only mild improvements in running time have been achieved so far \cite{vgf:13,akutsu1999,zakov2010}, and obtaining a truly sub-cubic algorithm for RNA-folding has remained open till this work.

Abboud et al.~\cite{abw:15} showed that obtaining an algorithm for RNA-folding that runs in $O(n^{\omega-\eps})$ time for any $\eps>0$ would result in a breakthrough for the Clique problem. Moreover, their results imply that any truly sub-cubic algorithm for RNA-folding must use fast matrix multiplication, unless there are fast algorithms for Clique that do not use fast matrix multiplication. Their results hold for alphabet $\Sigma$ of size $13$, which was recently improved to $|\Sigma|=2$~\cite{Chang16}.

\paragraph{Dyck LED} A problem closely related to RNA-folding is Dyck language edit distance, which is LED for the grammar of well-balanced parentheses. For example, $[()]$ belongs to the Dyck language, but $[)$ or $][$ do not. (The RNA grammar is often referred to as ``two-sided Dyck'', where $][$ is also a valid match.) Dyck edit distance with insertion and deletion generalizes the widely-studied string edit distance problem \cite{Masek:80,lms:98,bjkk:04,bfc:06,ak:09,ako:10}. When approximation is allowed, a near-linear time $O(\text{poly}\log{n})$-approximation algorithm was developed by Saha~\cite{s:14}. Moreover, a $(1+\epsilon)$-approximation in $O(n^\omega)$ time was shown in~\cite{s:15} for any constant $\epsilon > 0$. Abboud et al. \cite{abw:15} related the Dyck LED problem to Clique with the same implications as for RNA-folding. Thus, up to a breakthrough in Clique algorithms, truly sub-cubic Dyck LED requires fast matrix multiplication. Prior to our work, no sub-cubic exact algorithm was known for Dyck LED.

\subsection{Preliminaries and Notation}

In this paper, by ``randomized time $t(n)$'' we mean a zero-error randomized algorithm running in time $t(n)$ %in expectation, and also 
with high probability.\footnote{An event happens \emph{with high probability} (w.h.p.) if its probability is at least $1-1/n^{c}$ for some $c > 0$.}

\paragraph{Matrix Multiplication}
As is typical, we denote by $\omega<2.3729$~\cite{v12,Gall14a} the exponent of square matrix multiplication, i.e. $\omega$ is the infimum over all reals such that $n\times n$ matrix multiplication over the complex numbers can be computed in $n^{\omega+o(1)}$ time. 
For ease of notation and as is typical in the literature, we shall omit the $o(1)$ term and write $O(n^\omega)$ instead.
We denote the running time to multiply an $a \times b$ matrix with a $b \times c$ matrix by $\MM(a,b,c)$ \cite{le:12}. As in (1) above we have the following:

\begin{lemma}[\cite{Alon:1997}] \label{lem:matrixmult} Let $A,B$ be $a \times b$ and $b \times c$ matrices with entries in $\{-M,-M+1\ldots,M\} \cup \{\infty\}$. Then $A \minp B$ can be computed in time $\tilde O(M \cdot \MM(a,b,c))$. In particular, for $a=b=c=n$ this running time is $\tilde O(M n^\omega)$.
\end{lemma}

\paragraph{Context-Free Grammars and Scored Parsing} Let $G=(N,T,P,S)$ be a context-free grammar (CFG), where $N$ and $T$ are the (disjoint) sets of non-terminals and terminals, respectively, $P$ is the set of productions, and $S \in N$ is the start symbol. We recall that a production rule $p$ is of the form $X\rightarrow \alpha$, with $X\in N$ and\footnote{Given a set of symbols $U$, by $U^*$ we denote as usual any, possibly empty, string of elements from $U$.} $\alpha\in (N\cup T)^*$, and applying $p$ to (some instance of) $X\in N$ in a string $\sigma\in (N\cup T)^*$ generates the string $\sigma'$ where $X$ is replaced by $\alpha$. For $\alpha, \beta \in (N\cup T)^*$, we write $\alpha \to \beta$ if $\beta$ can be generated from $\alpha$ by applying one production rule, and we write $\alpha \to^* \beta$ (``$\beta$ can be derived from $\alpha$'') if there is a sequence of productions generating $\beta$ from $\alpha$.
The language $L(X)$ generated by a non-terminal $X \in N$ is the set of strings $\sigma\in T^*$ that can be derived from $X$. We also let $L(G) := L(S)$ denote the language generated by $G$.

At many places in this paper we may assume that $G$ is given in Chomsky normal form (CNF). Specifically, all productions are of the form $Z\rightarrow XY$, $Z\rightarrow c$, and $S\rightarrow \epsilon$, where $X,Y\in N\setminus \{S\}$, $Z\in N$, $c\in T$, and $\epsilon$ denotes the empty string.  

A \emph{scored grammar} is a CFG $G$ where each production rule $p\in P$ is associated with a non-negative integer score $s(p)$. Intuitively, applying production $p$ has a cost $s(p)$. The total score of any derivation is simply the sum of all scores of productions used in the derivation. For any $X \in N$ and $\sigma \in T^*$, we define $s_G(X,\sigma) = s(X,\sigma)$ as the minimum total score of any derivation $X \to^* \sigma$, or as $\infty$ if $\sigma \not\in L(X)$. The scored language generated by $X \in N$ is the set $\{ (\sigma, s(X,\sigma)) \mid \sigma \in L(X) \}$, and the scored language generated by $G$ is the scored language generated by the start symbol $S$. 
In the \emph{scored parsing problem} on grammar $G$, we are given a string $\sigma$ of length $n$, and we wish to compute $s(S,\sigma)$.

\paragraph{Organization} In Section~\ref{sec:minplus} we give our main technical result, a truly sub-cubic algorithm for the $(\min,+)$-product of BD matrices. In Section \ref{sec:improved}, we show how to further reduce the running time, how to derandomize our algorithm, and some generalizations of our approach.  In Section~\ref{sec:scoredParsing}, we show how bounded-difference scored parsing can be solved asymptotically in the same time as computing a single BD $(\min,+)$-product. Section \ref{sec:appl} is devoted to prove reductions from LED, RNA-folding, and OSG to scored parsing on BD grammars.

\section{Fast Bounded-Difference \boldmath$(\min,+)$-Product}
\label{sec:minplus}
In this section we present our fast algorithm for $(\min,+)$-product on BD matrices. For ease of presentation, we will focus here only on the case that both input matrices $A$ and $B$ are BD. Furthermore, we will  present a simplified randomized algorithm which is still truly sub-cubic. Refinements of the running time, derandomization, and generalizations are discussed in Section~\ref{sec:improved}. Let $A$ and $B$ be $n \times n$ matrices with $W$-bounded differences. We write $C = A\minp B$ for the desired output and denote by $\hat C$ the result computed by our algorithm. Our algorithm consists of the following three  main phases (see also Algorithm~\ref{alg:basic}).

\subsection{Phase 1: Computing an approximation}
Let $\Delta$ be a positive integer that we later fix as a small polynomial\footnote{We can assume that both $n$ and $\Delta$ are powers of two, so in particular we can assume that $\Delta$ divides $n$.} in $n$. We partition $[n]$ into blocks of length $\Delta$ by setting $I(i') := \{i \in [n] \mid i' - \Delta < i \le i' \}$ for any $i'$ divisible by $\Delta$. From now on by $i,k,j$ we denote indices in the matrices $A,B$, and $C$ and by $i',k',j'$ we denote numbers divisible by $\Delta$, i.e., indices of blocks. 

The first step of our algorithm is to compute an entry-wise additive $O(\Delta W)$-approximation $\tilde C$ of $A\minp B$. Since $A$ and $B$ are $W$-BD, it suffices to approximately evaluate $A \minp B$ only for indices $i',k',j'$ divisible by $\Delta$. Specifically, we compute $\tilde C_{i',j'} = \min\{A_{i',k'} + B_{k',j'} \mid k' \text{ divisible by } \Delta \}$, and set $\tilde C_{i,j} := \tilde C_{i',j'}$ for any $i \in I(i'), j \in I(j')$, see lines 1-3 of Algorithm~\ref{alg:basic}. 
The next lemma shows that $\tilde C$ is a good approximation of~$C$.

\begin{lemma} \label{lem:blockingerror}
  For any $i',k',j'$ divisible by $\Delta$ and any $(i,k,j) \in I(i') \times I(k') \times I(j')$ we have 
\begin{align*}
 (1) &\; |A_{i,k} - A_{i',k'}| \le 2 \Delta W, & (2) &\; |B_{k,j} - B_{k',j'}| \le 2 \Delta W,\\
(3) &\; |C_{i,j} - C_{i',j'}| \le 2 \Delta W, & (4) &\; |C_{i,j} - \tilde C_{i,j}| \le 4 \Delta W.
\end{align*}
\end{lemma}
\begin{proof}
Consider first (1). Observe that we can move from $A_{i,k}$ to $A_{i',k}$ in $i'-i\leq \Delta$ steps each time changing the absolute value by at most $W$, hence $|A_{i,k} - A_{i',k}| \le \Delta W$. Similarly, we can move from $A_{i',k}$ to $A_{i',k'}$. The overall absolute change is therefore at most $2\Delta W$. The proof of (2) is analogous.  
  
For (3), let $k$ be such that $C_{i,j} = A_{i,k} + B_{k,j}$. Then $C_{i',j'} \le A_{i',k} + B_{k,j'} \le A_{i,k} + B_{k,j} + 2 \Delta W = C_{i,j} + 2\Delta W$. In the second inequality we used the fact that $A_{i',k}\leq A_{i,k}+\Delta W$ and $B_{k,j'}\leq B_{k,j}+\Delta W$ from the same argument as above.  Symmetrically, we obtain $C_{i',j'} \le C_{i,j} + 2\Delta W$. 
  
  It remains to prove (4). Note that $\tilde C_{i,j} = \tilde C_{i',j'}$ by construction. Let $k'$ be divisible by $\Delta$ and such that $\tilde C_{i',j'} = A_{i',k'} + B_{k',j'}$. Then $C_{i,j} \le A_{i,k'} + B_{k',j} \le A_{i',k'} + B_{k',j'} + 2\Delta W = \tilde C_{i',j'} + 2\Delta W$, where again the second inequality exploits the above observation. For the other direction, let $k$ be such that $C_{i,j} = A_{i,k} + B_{k,j}$, and consider $k'$ with $k \in I(k')$. Then $\tilde C_{i',j'} \le A_{i',k'} + B_{k',j'} \le A_{i,k} + B_{k,j} + 4 \Delta W = C_{i,j} + 4 \Delta W$, where in the second inequality we exploited (1) and (2).
\end{proof}

\algnewcommand{\LineComment}[1]{\Statex \(\triangleright\) \emph{#1}}
\begin{algorithm}
\caption{$(\min,+)$-product $A\minp B$ for $n \times n$ matrices $A,B$ with $W$-bounded differences. Here $\Delta$ and $\rho$ are carefully chosen polynomial values. Also $I(q)=\{q-\Delta+1,\ldots,q\}$.}\label{alg:basic}
\begin{algorithmic}[1]
\LineComment{Phase 1: compute entry-wise additive 
$4\Delta W$-approximation $\tilde C$ of $A\minp B$}
\For {any $i',j'$ divisible by $\Delta$}  
  \State $\tilde C_{i',j'} := \min\{ A_{i',k'} + B_{k',j'} \mid k' \text{ divisible by } \Delta \}$ 
  \For {any $i \in I(i')$, $j \in I(j')$}
    $\tilde C_{i,j} := \tilde C_{i',j'}$
  \EndFor
\EndFor
\LineComment{Phase 2: randomized reduction to $(\min,+)$-product with small entries}
\State initialize all entries of $\hat C$ with $\infty$
\For {$1 \le r \le \rho$}
\State pick $i^r$ and $j^r$ independently and uniformly at random from $[n]$
\For {all $i,k$}
\State set $A^r_{i,k} := A_{i,k} + B_{k, j^r} - \tilde C_{i, j^r}$
\State if $A^r_{i,k} \not\in [-48\Delta W, 48\Delta W]$ then set $A^r_{i,k} := \infty$
\EndFor
\For {all $k,j$}
\State set $B^r_{k,j} := B_{k,j} - B_{k, j^r} + \tilde C_{i^r,j^r} - \tilde C_{i^r, j}$
\State if $B^r_{k,j} \not\in [-48\Delta W, 48\Delta W]$ then set $B^r_{k,j} := \infty$
\EndFor
\State compute $C^r := A^r \minp B^r$ using Lemma~\ref{lem:matrixmult} 
\For{all $i,j$} $\hat C_{i,j} := \min\{ \hat C_{i,j}, C^r_{i,j} + \tilde C_{i, j^r} - \tilde C_{i^r,j^r} + \tilde C_{i^r, j}\}$
\EndFor
\EndFor
\LineComment{Phase 3: exhaustive search over all relevant uncovered triples of indices}
\For {all $i',k',j'$ divisible by $\Delta$}
  \If{$|A_{i',k'} + B_{k',j'} - \tilde C_{i',j'}| \le 8\Delta W$} %\Comment{\emph{iterate over all relevant blocks}}
  \If{for all $r$ we have $|A^r_{i', k'}| > 44 \Delta W$ or $|B^r_{k', j'}| > 44 \Delta W$} 
  %\LineComment{\fab{Brute force computation over uncovered blocks}}
    \For{all $i \in I(i'), k \in I(k'), j \in I(j')$} 
      \State $\hat C_{i,j} := \min\{\hat C_{i,j}, A_{i,k} + B_{k,j}\}$
    \EndFor
  \EndIf
  \EndIf
\EndFor
\State \Return $\hat C$
\end{algorithmic}
\end{algorithm}

\subsection{Phase 2: Randomized reduction to \boldmath$(\min,+)$-product with small entries}
\label{sec:steptwo}

The second step of our algorithm is the most involved one. The goal of this step is to change $A$ and $B$ in a randomized way to obtain matrices where each entry is $\infty$ or has small absolute value, thus reducing the problem to Lemma~\ref{lem:matrixmult}. This step will cover most triples $i,k,j$, but not all: the third step of the algorithm will cover the remaining triples by exhaustive search. We remark that Phase 2 works with arbitrary matrices $A$ and $B$ (assuming we know an approximate answer $\tilde C$ as computed in Phase 1). 

The following observation is the heart of our argument. For any vector $F = (F_1,\ldots,F_n)$, adding $F_k$ to every entry $A_{i,k}$ ($\forall i$) and subtracting $F_k$ from every entry $B_{k,j}$ ($\forall j$) does not change the product $A\minp B$. Similarly, for $n$-dimension vectors $X$ and $Y$, adding $X_i$ to every entry $A_{i,k}$ and adding $Y_j$ to every entry $B_{k,j}$ changes the entry $(A\minp B)_{i,j}$ by $+X_i + Y_j$, which we can cancel after computing the product. 

Specifically, we may fix indices $i^r, j^r$ and consider the matrices $A^r$ with $A^r_{i,k} := A_{i,k} + B_{k, j^r} - \tilde C_{i, j^r}$ and $B^r$ with $B^r_{k,j} := B_{k,j} - B_{k, j^r} + \tilde C_{i^r,j^r} - \tilde C_{i^r, j}$. Then from $C^r := A^r \minp B^r$ we can infer $C = A\minp B$ via the equation $C_{i,j} = C^r_{i,j} + \tilde C_{i, j^r} - \tilde C_{i^r,j^r} + \tilde C_{i^r, j}$. 

We will set an entry of $A^r$ or $B^r$ to $\infty$ if its absolute value is more than $48 \Delta W$. 
%This allows to compute $C^r = A^r \minp B^r$ efficiently in time $\tilde O(\Delta W n^\omega)$ by Lemma~\ref{lem:matrixmult}. 
This allows us to compute $C^r = A^r \minp B^r$ efficiently using Lemma~\ref{lem:matrixmult}. However, it does not correctly compute $C = A\minp B$. Instead, we obtain values $\hat C_{i,j}^r := C^r_{i,j} + \tilde C_{i,j^r} - \tilde C_{i^r,j^r} + \tilde C_{i^r,j}$ that fulfill $\hat C_{i,j}^r \ge C_{i,j}$. Moreover, if neither $A^r_{i,k}$ nor $B^r_{k,j}$ was set to $\infty$ then $\hat C_{i,j}^r \le A_{i,k} + B_{k,j}$; in this case the contribution of $i,k,j$ to $C_{i,j}$ is incorporated in $\hat C_{ij}^r$ (and we say that $i,k,j$ is ``covered'' by $A^r, B^r$, see Definition~\ref{def:weakstrongnotions}). We repeat this procedure with independently and uniformly random $i^r, j^r \in [n]$ for $r=1,\ldots,\rho$ many rounds, where $1 \le \rho \le n$ is a small polynomial in $n$ to be fixed later. Then $\hat C$ is set to the entry-wise minimum over all $\hat C^r$. This finishes the description of Phase~2, see lines 4--14 of Algorithm~\ref{alg:basic}.

In the analysis of this step of the algorithm, we want to show that w.h.p.\ most of the ``relevant'' triples $i,k,j$ get covered: in particular, all
triples with $A_{i,k} + B_{k,j} = C_{i,j}$ are relevant, as these triples define the output. However, since this definition would depend on the output $C_{i,j}$, we can only (approximately) check a weak version of relevance, see Definition~\ref{def:weakstrongnotions}. Similarly, we need a weak version of being covered.

\begin{definition} \label{def:weakstrongnotions}
  We call a triple $(i,k,j)$ 
  \begin{itemize}\itemsep0pt
    \item \emph{strongly relevant} if $A_{i,k} + B_{k,j} = C_{i,j}$, 
    \item \emph{weakly relevant} if $|A_{i,k} + B_{k,j} - C_{i,j}| \le 16 \Delta W$,
    \item \emph{strongly $r$-uncovered} if for all $1 \le r' \le r$ we have $|A^{r'}_{i,k}| > 48 \Delta W$ or $|B^{r'}_{k,j}| > 48 \Delta W$, and
    \item \emph{weakly $r$-uncovered} if for all $1 \le r' \le r$ we have $|A^{r'}_{i,k}| > 40 \Delta W$ or $|B^{r'}_{k,j}| > 40 \Delta W$.
  \end{itemize}
A triple is strongly (resp., weakly) uncovered if it is strongly (resp., weakly) $\rho$-uncovered. Finally, a triple is strongly (resp., weakly) $r$-covered if it is not strongly (resp., weakly) $r$-uncovered.
\end{definition}

The next lemma gives a sufficient condition for being weakly $r$-covered.
  \begin{lemma} \label{lem:fourcycleargument}
   For any $i,k,j$ and $i^r, j^r$, if all triples $(i,k,j^r)$, $(i^r,k,j^r)$, $(i^r,k,j)$ are weakly relevant then $(i,k,j)$ is weakly $r$-covered.
  \end{lemma}
  \begin{proof}
  From the assumption and $\tilde C$ being an additive $4 \Delta W$-approximation of $C$, we obtain
  $$ |A_{i,k} + B_{k,j^r} - \tilde C_{i,j^r}| \le |A_{i,k} + B_{k,j^r} - C_{i, j^r}| + |\tilde C_{i, j^r} - C_{i, j^r}| \le 16 \Delta W + 4 \Delta W = 20 \Delta W. $$
  Similarly, we also have $|A_{i^r,k} + B_{k,j^r} - \tilde C_{i^r, j^r}| \le 20 \Delta W$ and $|A_{i^r,k} + B_{k,j} - \tilde C_{i^r, j}| \le 20 \Delta W$.
  
Recall that in the algorithm we set $A^r_{i,k} := A_{i,k} + B_{k,j^r} - \tilde C_{i,j^r}$ and $B^r_{k,j} := B_{k,j} - B_{k, j^r} + \tilde C_{i^r,j^r} - \tilde C_{i^r, j}$ (and then reset them to $\infty$ if their absolute value is more than $48 \Delta W$). From the above inequalities, we have $|A^r_{i,k}| \le 20 \Delta W$. Moreover, we can write $B^r_{k,j}$ as $(A_{i^r, k} + B_{k,j} - \tilde C_{i^r, j}) - (A_{i^r, k} + B_{k, j^r} - \tilde C_{i^r, j^r})$, where both terms in brackets have absolute value bounded by $20 \Delta W$, and thus $|B^r_{k,j}| \le 40 \Delta W$. It follows that the triple $i,k,j$ gets weakly covered in round~$r$. 
\end{proof}

We will crucially exploit the following well-known extremal graph-theoretic result~\cite{erdos1984cube,erdHos1983supersaturated}. We present the easy proof for completeness.

\begin{lemma} \label{lem:countfourcycles}
Let $G=(U\cup V,E)$ be a bipartite graph with $|U|=|V|=n$ nodes per partition and $|E|=m$ edges. Let $C$ be the number of $4$-cycles of $G$. If $m\geq 2n^{3/2}$, then   $C \geq m^4/(32 n^4)$.
\end{lemma}
\begin{proof}
For any pair of nodes $v,v' \in V$, let $N(v,v')$ be the number of common neighbors $\{u\in U \mid \{u,v\},\{u,v'\}\in E\}$, and let $N=\sum_{\{v,v'\} \in {V \choose 2}} N(v,v')$. By $d(w)$ we denote the degree of node $w$ in~$G$.
By convexity of ${x \choose 2} = \frac{x(x-1)}{2}$ and Jensen's inequality, we have
$$
N=\sum_{\{v,v'\} \in {V \choose 2}} N(v,v')= \sum_{u\in U} {d(u)\choose 2} \ge n \cdot {\sum_{u \in U} d(u)/n \choose 2} = n {m/n \choose 2} = \frac{m^2}{2n} - \frac m2 \ge \frac{m^2}{2n} - n^2.
$$
Since $m \ge 2n^{3/2}$ by assumption, we derive $\frac{m^2}{2n} \ge 2 n^2$ and thus we obtain $N\geq n^2 > 2{n\choose 2}$ as well as $N\geq m^2/(4n)$.

By the same convexity argument as above, we also have
$$
C = \sum_{\{v,v'\} \in {V \choose 2}} {N(v,v')\choose 2} \geq {n\choose 2} \cdot {N/{n\choose 2} \choose 2}= \bigg(N-{n\choose 2}\bigg) \frac{N}{n(n-1)} \geq \frac{N^2}{2n^2},
$$
where in the last inequality above we used the fact that $N\ge  2{n\choose 2}$. Altogether, this yields 
$$
C \geq \frac{N^2}{2n^2} \geq \frac{m^4/(16n^2)}{2n^2} = \frac{m^4}{32 n^4},
$$
finishing the proof.
\end{proof}

We are now ready to lower bound the progress made by the algorithm at each round. 

\begin{lemma} \label{lem:numuncovered-rho}
  W.h.p.\ for any $\rho\geq 1$ the number of weakly relevant, weakly uncovered triples is $\tilde O(n^{2.5}+n^3 / \rho^{1/3})$.
\end{lemma}
\begin{proof}
%  We first show a sufficient condition for not being weakly $r$-uncovered.
Fix $k \in [n]$.
We construct a bipartite graph $G_{k}$ on $n+n$ vertices (we denote vertices in the left vertex set by $i$ or $i^r$ and vertices in the right vertex set by $j$ or $j^r$). We add edge $(i,j)$ to $G_{k}$ if the triple $(i,k,j)$ is weakly relevant. We also consider the subgraph $G'_k$ of $G_k$ containing edge $(i,j)$ if and only if $(i,k,j)$ is weakly relevant and weakly uncovered.

Let $z= c(n^2/\rho) \ln n$ for any constant $c>3$. Consider an edge $(i,j)$ in $G_k$ that is contained in at least~$z$ $4$-cycles. Now consider 
each round $r$ in turn and let $i\rightarrow \ell\rightarrow p\rightarrow j\rightarrow i$ be a $4$-cycle containing $(i,j)$. 
If $i^r=p$ and $j^r=\ell$ are selected, then since, by the definition of $G_k$, $(i,k,\ell),(p,k,\ell)$ and $(p,k,j)$ are weakly relevant, by Lemma~\ref{lem:fourcycleargument}, $(i,k,j)$ will be $r$-covered and thus $(i,j)$ is not an edge in $G'_k$.

Thus, if in any round $r$ the indices $i^r,j^r$ are selected to be among the at least $z$ choices of vertices that complete $(i,j)$ to a $4$-cycle in $G_k$, then $(i,j)$ is not in $G'_k$. For a particular edge $(i,j)$ with at least $z$ $4$-cycles in a particular $G_k$, the probability that $i^r,j^r$ are {\em never} picked to form a $4$-cycle with $(i,j)$  is 
$$\leq \left(1-\frac{z}{n^2}\right)^{\rho} = \left(1-\frac{z}{n^2}\right)^{c(n^2/z) \ln n}\leq \frac{1}{n^c}.$$
By a union bound, over all $i,j,k$ we obtain an error probability of at most $1/n^{c-3}$, which is $1/\poly(n)$ as we picked $c>3$.
Hence, with high probability every edge in every $G'_k$ is contained in less than $z$ $4$-cycles in $G_k$.

Let $m_k$ denote the number of edges of $G'_k$. Since w.h.p.\ every edge in $G'_k$ is contained in less than~$z$ $4$-cycles in $G_k$ (and thus also in $G'_k$), the number of $4$-cycles $C(k)$ of $G'_k$ is less than $m_kz$. On the other hand, by Lemma~\ref{lem:countfourcycles}, we have $m_k < 2 n^{3/2}$ or $C(k)\geq (m_k/n)^4/32$. In the latter case, we obtain 
\[(m_k/n)^4<32 m_k z \implies m_k^3< 32 n^4 z \implies m_k < \left(32c(n^6/\rho) \ln n\right)^{1/3} \implies m_k\leq \tilde{O}(n^2/\rho^{1/3}). \]
Together, this yields $m_k = \tilde O(n^{1.5}+n^2/\rho^{1/3})$.
Finally, note that the number of weakly relevant, weakly uncovered triples is $\sum_k m_k = \tilde O(n^{2.5}+n^3/\rho^{1/3})$.
  \end{proof}

\subsection{Phase 3: Exhaustive search over all relevant uncovered triples of indices}
\label{sec:stepthree}

In the third and last phase we make sure to fix all strongly relevant, strongly uncovered triples by exhaustive search, as these are the triples defining the output matrix whose contribution is not yet incorporated in $\hat C$. We are allowed to scan all weakly relevant, weakly uncovered triples, as we know that their number is small by Lemma \ref{lem:numuncovered-rho}. 
This is the only phase that requires that $A$ and $B$ are BD.

We use the following definitions of being \emph{approximately} relevant or uncovered, since they are identical for all triples $(i,k,j)$ in a block $i',k',j'$ and thus can be checked efficiently. 

\begin{definition} \label{def:approxnotions}
  We call a triple $(i,k,j) \in I(i') \times I(k') \times I(j')$ 
  \begin{itemize}
    \item \emph{approximately relevant} if $|A_{i',k'} + B_{k',j'} - \tilde C_{i',j'}| \le 8 \Delta W$, and
    \item \emph{approximately $r$-uncovered} if for all $1 \le r' \le r$ we have $|A^{r'}_{i', k'}| > 44 \Delta W$ or $|B^{r'}_{k', j'}| > 44 \Delta W$.
  \end{itemize}
  A triple is approximately uncovered if it is approximately $\rho$-uncovered. %Finally, a triple is approximately $r$-covered if it is not approximately $r$-uncovered.
\end{definition}

The notions of being strongly, weakly, and approximately relevant/uncovered are related as follows.
\begin{lemma} \label{lem:relatenotions}
  Any strongly relevant triple is also approximately relevant. Any approximately relevant triple is also weakly relevant. The same statements hold with ``relevant'' replaced by ``$r$-uncovered''.
\end{lemma}

\begin{proof}
Let $(i,k,j) \in I(i') \times I(k') \times I(j')$. Using Lemma~\ref{lem:blockingerror}, we can bound the absolute difference between $A_{i,k} + B_{k,j} - C_{i,j}$ and $A_{i',k'} + B_{k',j'} - \tilde C_{i',j'}$ by the three contributions $|A_{i,k} - A_{i',k'}| \le 2 \Delta W$, $|B_{k,j} - B_{k',j'}| \le 2 \Delta W$, and $|C_{i,j} - \tilde C_{i',j'}| = |C_{i,j} - \tilde C_{i,j}| \le 4 \Delta W$. Thus, if $A_{i,k} + B_{k,j} = C_{i,j}$ (i.e., $(i,k,j)$ is strongly relevant), then $|A_{i',k'} + B_{k',j'} - \tilde C_{i',j'}| \le 8 \Delta W$ (i.e., $(i,k,j)$ is approximately relevant). On the other hand, if $(i,k,j)$ is approximately relevant, then $|A_{i,k} + B_{k,j} - C_{i,j}| \le 16 \Delta W$ (i.e., $(i,k,j)$ is weakly relevant).
  
  For the notion of being $r'$-uncovered, for any $1 \le r \le r'$ we bound the absolute differences $|A_{i,k}^r - A_{i',k'}^r|$ and $|B_{k,j}^r - B_{k',j'}^r|$. Recall that we set $A_{i,j}^r := A_{i,j} + B_{k,j^r} - \tilde C_{i,j^r}$. Again using Lemma~\ref{lem:blockingerror}, we bound both $|A_{i,j} - A_{i',j'}|$ and $|B_{k, j^r} - B_{k', j^r}|$ by $2 \Delta W$. Since we have $\tilde C_{i, j^r} = \tilde C_{i', j^r}$ by definition, in total we obtain $|A_{i,k}^r - A_{i',k'}^r| \le 4 \Delta W$. Similarly, recall that we set $B^r_{k,j} := B_{k,j} - B_{k, j^r} + \tilde C_{i^r,j^r} - \tilde C_{i^r, j}$. The first two terms both contribute at most $2 \Delta W$, while the latter two terms are equal for $B_{k,j}^r$ and $B_{k',j'}^r$. Thus, $|B_{k,j}^r - B_{k',j'}^r| \le 4 \Delta W$. The statements on ``$r$-uncovered'' follow immediately from these inequalities.
\end{proof}

In our algorithm, we enumerate every triple $(i',k',j')$ whose indices are divisible by $\Delta$, and check whether that triple is approximately relevant. Then we check whether it is approximately uncovered. If so, we perform an exhaustive search over the block $i',k',j'$: We iterate over all $(i,k,j) \in I(i') \times I(k') \times I(j')$ and update $\hat C_{i,j} := \min\{\hat C_{i,j}, A_{i,k} + B_{k,j}\}$, see lines 15-19 of Algorithm~\ref{alg:basic}. 

Note that $i',k',j'$ is approximately relevant (resp., approximately uncovered) if and only if all $(i,k,j) \in I(i') \times I(k') \times I(j')$ are approximately relevant (resp., approximately uncovered). Hence, we indeed enumerate all approximately relevant, approximately uncovered triples, and by Lemma~\ref{lem:relatenotions} this is a superset of all strongly relevant, strongly uncovered triples. Thus, every strongly relevant triple $(i,k,j)$ contributes to $\hat C_{i,j}$ in Phase 2 or Phase 3. This proves correctness of the output matrix $\hat C$.

\subsection{Running Time} \label{sec:timeanalysis}

The running time of Phase 1 is $O((n/\Delta)^3 + n^2)$ using brute-force. The running time of Phase 2 is $\tilde O(\rho \Delta W n^\omega)$, since there are $\rho$ invocations of Lemma~\ref{lem:matrixmult} on matrices whose finite entries have absolute value $O(\Delta W)$.
It remains to consider Phase $3$. Enumerating all blocks $i',k',j'$ and checking whether they are approximately relevant and approximately uncovered takes time $O((n/\Delta)^3 \rho)$. The approximately relevant and approximately uncovered triples form a subset of the weakly relevant and weakly uncovered triples by Lemma~\ref{lem:relatenotions}. The number of the latter triples is upper bounded by $\tilde O(n^{2.5}+n^3 / \rho^{1/3})$ w.h.p.\ by Lemma \ref{lem:numuncovered}. Thus, w.h.p.\ Phase 3 takes total time $\tilde O((n/\Delta)^3 \rho + n^3 / \rho^{1/3}+n^{2.5})$. In total, the running time of Algorithm \ref{alg:basic} is w.h.p. 
$$
\tilde O((n/\Delta)^3 + n^2 + \rho \Delta W n^\omega + (n/\Delta)^3 \rho + n^3 / \rho^{1/3}+n^{2.5}).
$$ 

A quick check shows that for appropriately chosen $\rho$ and $\Delta$ (say $\rho := \Delta := n^{0.1}$) and for sufficiently small $W$ this running time is truly sub-cubic. We optimize by setting $\rho := (n^{3-\omega}/W)^{9/16}$ and $\Delta := (n^{3-\omega}/W)^{1/4}$, obtaining time $\tilde O(W^{3/16} n^{(39+3\omega)/16})$, which is truly sub-cubic for $W \le O(n^{3-\omega-\eps})$. For $W=O(1)$ using $\omega \le 2.3729$~\cite{v12,Gall14a} this running time evaluates to $O(n^{2.8825})$.

\section{Bounded-Difference \boldmath$(\min,+)$-Product: Improvements, Derandomization, and Generalizations}
\label{sec:improved}

In this section, we prove Theorem~\ref{thr:mainProduct} by improving on the running time from Section~\ref{sec:minplus}.
\subsection{Speeding Up Phase 2}

We begin with a more refined version of Lemma~\ref{lem:numuncovered-rho}. Recall that $\rho$ is the maximum number of iterations in Phase $2$.

\begin{lemma} \label{lem:numuncovered}
  W.h.p.\ for any $1 \le r \le \rho$ the number of weakly relevant, weakly $r$-uncovered triples is $\tilde O(n^{2.5}+n^3 / r^{1/3})$.
\end{lemma}
\begin{proof}
  We first show a sufficient condition for not being weakly $r$-uncovered.
Fix $k \in [n]$. For any $1 \le r \le \rho+1$, we construct a bipartite graph $G_{r,k}$ on $n+n$ vertices (we denote vertices in the left vertex set by $i$ or $i^r$ and vertices in the right vertex set by $j$ or $j^r$). We add edge $\{i,j\}$ to $G_{r,k}$ if the triple $(i,k,j)$ is weakly relevant and weakly $(r-1)$-uncovered. Note that $E(G_{r,k}) \supseteq E(G_{r',k})$ for $r \le r'$.
    Denote the number of edges in $G_{r,k}$ by $m_{r,k}$ and its density by $\alpha_{r,k} = m_{r,k} / n^2$. In the following we show that  as a function of $r$ the number of edges~$m_{r,k}$ drops by a constant factor after $O(\alpha_{r,k}^{-3} \log(n))$ rounds w.h.p., as long as the density is large enough.

    We denote by $C_{r,k}(i,j)$ the number of 4-cycles in $G_{r,k}$ containing edge $\{i,j\}$. (If $\{i,j\}$ is not an edge in $G_{r,k}$, we set $C_{r,k}(i,j) = 0$.) Observe that $C_{r,k}(i,j) \ge C_{r',k}(i,j)$ for $r \le r'$. 
    
Now fix a round $r$. For $r \le r'$, we call $\{i,j\}$ \emph{$r'$-heavy} if $C_{r',k}(i,j) \ge 2^{-8} \alpha_{r,k}^3 n^2$. 
    Let $r^*$ be a round with $r^* - r = \Theta( \alpha_{r,k}^{-3} \log n )$ (with sufficiently large hidden constant). 
    We claim that w.h.p.\ no $\{i,j\}$ is $r^*$-heavy. Indeed, in any round $r \le r' < r^*$, either $\{i,j\}$ is not $r'$-heavy, say because some of the edges in its 4-cycles got covered in the last round, but then we are done. Or $\{i,j\}$ is $r'$-heavy, but then with probability $C_{r',k}(i,j) / n^2 = \Omega(\alpha_{r,k}^{3})$ we choose $i^{r'},j^{r'}$ as the remaining vertices in one of the 4-cycles containing $\{i,j\}$. In this case, Lemma~\ref{lem:fourcycleargument} shows that $(i,k,j)$ will get weakly covered in round $r'$, so in particular $\{i,j\}$ is not $(r'+1)$-heavy. Over $r^* - r = \Theta( \alpha_{r,k}^{-3} \log n )$ rounds, this event happens with high probability.
    
    Now we know that w.h.p.\ no $\{i,j\}$ is $r^*$-heavy. Thus, each of the $\alpha_{r^*,k} n^2$ edges of $G_{r^*,k}$ is contained in less than $2^{-8} \alpha_{r,k}^3 n^2$ 4-cycles, so that the total number of 4-cycles in $G_{r^*,k}$ is at most $2^{-8} \alpha_{r^*,k} \alpha_{r,k}^3 n^4$. On the other hand, Lemma~\ref{lem:countfourcycles} shows that the number of 4-cycles is at least $(\alpha_{r^*,k} n^2)^4 / (32 n^4)$ if $\alpha_{r^*,k} \ge 2/\sqrt{n}$. Altogether, we obtain $\alpha_{r^*,k} \le \max\{\alpha_{r,k} / 2, 2/\sqrt{n}\}$. In particular, w.h.p.\ in round $r = O(\sum_{i=0}^{t} 2^{3i} \log n) = O(2^{3t} \log n)$ the density of $G_{r,k}$ is at most $2^{-t}$, as long as $2^{-t} \ge 2/\sqrt{n}$. In other words, w.h.p.\ the density of $G_{r,k}$ is $O( (\log(n)/r)^{1/3} + n^{-1/2} )$, and $m_{r,k} \le O(n^2 (\log(n)/r)^{1/3} + n^{3/2} )$. 
    Since $m_{r+1,k}$ counts the weakly relevant, weakly $r$-uncovered triples $(i,k,j)$ for fixed $k$, summing over all $k \in [n]$ yields the claim.
\end{proof}

Inspection of the proof of Lemma~\ref{lem:numuncovered} shows that we only count triples $i,k,j$ that get covered in round $r$ if the triple $i^r,k,j^r$ is weakly relevant and weakly $(r-1)$-uncovered. Hence, after line 12 of Algorithm~\ref{alg:basic} we can remove all columns $k$ from $A^r$ and all rows $k$ from $B^r$ for which $i^r,k,j^r$ is not weakly relevant or not weakly $(r-1)$-uncovered. Then Lemma~\ref{lem:numuncovered} still holds, so the other steps are not affected. Note that checking this property for $i^r,k,j^r$ takes time $O(\rho)$ for each $k$ and each round $r$, and thus in total incurs cost $O(n \rho^2) \le O(\rho n^2)$, which is dominated by the remaining running time of Phase 2.
Using rectangular matrix multiplication to compute $A^r * B^r$ (Lemma~\ref{lem:matrixmult}) we obtain the following improved running time.

\begin{lemma}
  W.h.p.\ the improved Step 2 takes time $\tilde O(\rho \Delta W \cdot \MM(n,n /\rho^{1/3},n))$.
\end{lemma}
\begin{proof}
  Let $s_r$ denote the number of surviving $k$'s in round $r$, i.e., the number of $k$ such that $i^r,k,j^r$ is weakly relevant, weakly $(r-1)$-uncovered. Using Lemma~\ref{lem:matrixmult}, the running time of Step 2 is bounded by $\tilde O\big(\sum_{r=1}^\rho \Delta W \cdot \MM(n, s_r, n) \big)$. Note that for any $x,y$, we have $\MM(n,x,n) \le O((1+ x/y) \MM(n,y,n))$, by splitting columns and rows of length $x$ into $\lceil x/y \rceil \le 1+x/y$ blocks. Hence, we can bound the running time by $\tilde O\big(\sum_{r=1}^\rho \Delta W \cdot (1+s_r \rho^{1/3} / n) \cdot \MM(n, n / \rho^{1/3}, n) \big)$. Thus, to show the desired bound of $\tilde O(\rho \Delta W \cdot \MM(n, n / \rho^{1/3}, n))$, it suffices to show that $\sum_{r=1}^\rho s_r \le \tilde O(n \rho^{2/3})$ holds w.h.p.
  
  W.h.p.\ the number of weakly relevant, weakly $(r-1)$-uncovered triples is $\tilde O(n^3 / r^{1/3})$, by Lemma~\ref{lem:numuncovered}. Thus, for a random $k$ the probability that $i^r,k,j^r$ is weakly relevant, weakly $(r-1)$-uncovered is $\tilde O(r^{-1/3})$. Summing over all $k$ we obtain $\Ex[s_r] = \tilde O(n / r^{1/3})$ (note that the inequality $s_r \le n$ allows us to condition on any w.h.p.\ event for evaluating the expected value). This yields the desired bound for the expectation of the running time, since $\sum_{r=1}^\rho \Ex[s_r] \le \tilde O( n \sum_{r=1}^\rho r^{-1/3}) \le \tilde O(n \rho^{2/3})$.
  
  For concentration, fix $r^*$ as any power of two and consider $s_{r^*} + s_{r^*+1} + \ldots + s_{2r^*-1}$. For any $r^* \le r < 2r^*$ denote by $\bar s_{r}$ the number of triples $i^{r},k,j^{r}$ that are weakly relevant and weakly $r^*$-uncovered, and note that $s_{r} \le \bar s_{r}$. Again we have $\Ex[\bar s_{r}] \le \tilde O(n / r^{1/3})$. Moreover, conditioned on the choices up to round $r^*$, the numbers $\bar s_{r}$, $r^* \le r < 2r^*$, are independent. Hence, a Chernoff bound (Lemma~\ref{lem:chernoff} below) on variables $\bar s_r / n \in [0,1]$ shows that w.h.p.
  $$\bar s_{r^*} + \bar s_{r^*+1} + \ldots + \bar s_{2r^*-1} \le O\big(\Ex[\bar s_{r^*} + \bar s_{r^*+1} + \ldots + \bar s_{2r^*-1}] + n \log n\big).$$ 
  Hence, w.h.p.\ $\sum_{r=1}^\rho s_r \le \sum_{r=1}^\rho \bar s_r \le O\big( n \log(n) \log(\rho) + \sum_{r=1}^\rho \Ex[\bar s_r] \big)$. Using our bound on $\Ex[\bar s_{r}]$, we obtain w.h.p.\ $\sum_{r=1}^\rho s_r \le \tilde O(n + n \rho^{2/3} ) \le \tilde O(n \rho^{2/3})$ as desired.
\end{proof}
  
  \begin{lemma} \label{lem:chernoff}
    Let $X_1,\ldots,X_n$ be independent random variables taking values in $[0,1]$, and set $X := \sum_{i=1}^n X_i$. Then for any $c \ge 1$ we have
    $$ \Pr[X > (1+6ec) \Ex[X] + c \log n] \le n^{-c}. $$
  \end{lemma}
  \begin{proof}
    If $\Ex[X] < \log(n) / 2e$ we use the standard Chernoff bound $\Pr[X > t] \le 2^{-t}$ for $t > 2e \Ex[X]$ with $t := c \log n$. 
    Otherwise, we use the standard Chernoff bound $\Pr[X > (1+\delta) \Ex[X]] \le \exp(-\delta \Ex[X]/3)$ for $\delta \ge 1$ with $\delta := 6 e c$. 
  \end{proof}

\subsection{Speeding Up Phase 3}

\paragraph{Enumerating approximately uncovered blocks}
In line 17 of Algorithm~\ref{alg:basic} we check for each block $i',k',j'$ of approximately relevant triples whether it consists of approximately uncovered triples. This step can be improved using rectangular matrix multiplication as follows. For each block $k'$ we construct a $(n/\Delta) \times \rho$ matrix $U^{k'}$ and a $\rho \times (n/\Delta)$ matrix $V^{k'}$ with entries $U^{k'}_{x r} := [ |A^r_{x \Delta, k'}| \le 44 \Delta W ]$ and $V^{k'}_{r y} := [ |B^r_{k', y \Delta}| \le 44 \Delta W ]$. Then from the Boolean matrix product $U^{k'} \cdot V^{k'}$ we can infer for any block $i',k',j'$ whether it consists of approximately uncovered triples by checking $( U^{k'} \cdot V^{k'} )_{i'/\Delta, j'/\Delta} = 1$. Hence, enumerating the approximately relevant, approximately uncovered triples $i',k',j'$ can be done in time $O( (n/\Delta) \cdot \MM(n/\Delta, \rho, n/\Delta) )$. %, where $\MM(a,b,c)$ is the time to multiply an $a \times b$ matrix with a $b \times c$ matrix.

\paragraph{Recursion}
In the exhaustive search in Step 3, see lines 18-19 of Algorithm~\ref{alg:basic}, we essentially compute the $(\min,+)$-product of the matrices $(A_{ik})_{i \in I(i'), k \in I(k')}$ and $(B_{kj})_{k \in I(k'), j \in I(j')}$. These matrices again have $W$-BD, so we can use Algorithm~\ref{alg:basic} recursively to compute their product. Writing $T(n,W)$ for the running time of our algorithm, this reduces the time complexity of one invocation of lines 18-19 from $O(\Delta^3)$ to $T(\Delta,W)$, which in total reduces the running time of the exhaustive search from $\tilde O(n^3 /\rho^{1/3})$ to $\tilde O((T(\Delta,W)/\Delta^3) \cdot n^3 /\rho^{1/3})$ w.h.p. 

\subsection{Total running time} \label{sec:improvedtimeanalysis}

Recall that Step 1 takes time $O((n/\Delta)^3 + n^2)$, Step 2 now runs in $\tilde O(\rho \Delta W \cdot \MM(n,n / \rho^{1/3},n))$ w.h.p., and Step 3 now runs in $\tilde O( (n/\Delta) \cdot \MM(n/\Delta, \rho, n/\Delta) +  (T(\Delta,W)/\Delta^3) \cdot n^3 /\rho^{1/3})$ w.h.p. This yields the complicated recursion
$$ T(n,W) \le \tilde O\bigg( \rho \Delta W \cdot \MM\Big(n, \frac{n}{\rho^{1/3}},n\Big) + \frac n \Delta \cdot \MM\Big(\frac n \Delta, \rho, \frac n \Delta \Big) +  \frac{T(\Delta,W)}{\Delta^3} \cdot  \frac{n^3}{\rho^{1/3}} \bigg), $$
while the trivial algorithm yields $T(n,W) \le O(n^3)$.

In the remainder, we focus on the case $W = O(1)$, so that $T(n,W) = T(n,O(1)) =: T(n)$. Setting $\Delta := n^{\delta}$ and $\rho := n^{s} \log^c n$ for constants $\delta,s \in (0,1)$ and sufficiently large $c>0$, and using $\MM(a,\tilde O(b),c) \le \tilde O(\MM(a,b,c))$, we obtain 
$$T(n) \le \tilde O\big( n^{\delta+s} \MM(n, n^{1-s/3}, n)  + n^{1-\delta} \MM(n^{1-\delta}, n^s, n^{1-\delta}) \big) + n^{3-3\delta-s/3} T(n^\delta). $$
This is a recursion of the form $T(n) \le \tilde O( n^{\alpha} ) + n^\beta T(n^\gamma)$, which solves to $T(n) \le \tilde O( n^\alpha + n^{\beta / (1-\gamma)})$, by an argument similar to the master theorem. Hence, we obtain
$$ T(n) \le \tilde O\big( n^{\delta+s} \MM(n,n^{1-s/3},n) + n^{1-\delta} \MM(n^{1-\delta},n^s,n^{1-\delta}) + n^{(3-3\delta-s/3)/(1-\delta)} \big).$$
We optimize this expression using the bounds on rectangular matrix multiplication by Le Gall~\cite{le:12}. Specifically, we set $\delta := 0.0772$ and $s := 0.4863$ to obtain a bound of $O(n^{2.8244})$, which proves part of Theorem~\ref{thr:mainProduct}. 
Here we use the bounds $\MM(m,m^{1-s/3},m) \le \MM(m,m^{0.85},m) \le O(m^{2.260830})$ and $\MM(m,m^{s/(1-\delta)},m) \le \MM(m,m^{0.5302},m) \le O(m^{2.060396})$ by Le Gall~\cite{le:12} for $m = n$ and $m = n^{1-\delta}$, respectively.

We remark that if perfect rectangular matrix multiplication exists, i.e., $\MM(a,b,c) = \tilde O(ab + bc + ac)$, then our running time becomes $T(n) \le \tilde O( n^{2+\delta+s} + n^{3-3\delta} + n^{(3-3\delta-s/3)/(1-\delta)})$, which is optimized for $\delta = (13-\sqrt{133})/18$ and $s = (2\sqrt{133}-17)/9$, yielding an exponent of $(5+\sqrt{133})/6 \approx 2.7554$. This seems to be a barrier for our approach.

\subsection{Derandomization}

The only random choice in Algorithm~\ref{alg:basic} is to pick $i^r,j^r$ uniformly at random from $[n]$. In the following we show how to derandomize this choice, at the cost of increasing the running time of Step 2 by $O(\rho (n/\Delta)^{1+\omega})$. We then show that we still obtain a truly sub-cubic total running time.

Fix round $r$. Similar to the proof of Lemma~\ref{lem:numuncovered}, for any $k'$ divisible by $\Delta$ we construct a bipartite graph $G'_{r,k'}$ with vertex sets $\{\Delta, 2\Delta, \ldots, n\}$ and $\{\Delta, 2\Delta, \ldots, n\}$ (we denote vertices in the left vertex set by $i'$ or $i^r$ and vertices in the right vertex set by $j'$ or $j^r$). We connect $i',j'$ by an edge in $G'_{r,k'}$ if $i',k',j'$ is approximately relevant and approximately $(r-1)$-uncovered. In $G'_{r,k'}$ we count the number of 3-paths between any $i',j'$. Now we pick $i^r,j^r$ as the pair $i',j'$ maximizing the sum over all $k'$ of the number of 3-paths in $G'_{r,k'}$ containing $i',j'$. This finishes the description of the adapted algorithm.

It is easy to see that this adaptation of the algorithm increases the running time of Step 2 by at most $O(\rho (n/\Delta)^{\omega+1})$. Indeed, constructing all graphs $G'_{r,k'}$ over the $\rho$ rounds takes time $O(\rho (n/\Delta)^3)$, and computing the number of 3-paths between any pair of vertices can be done in $O(|V(G'_{r,k'})|^\omega)$, which over all $r$ and $k'$ incurs a total cost of $O(\rho (n/\Delta)^{\omega+1})$.

It remains to argue that an analog of Lemma~\ref{lem:numuncovered} still holds. Note that the number of 3-paths in $G'_{r,k'}$ containing $i^r,j^r$ counts the number of $i',j'$ such that $(i',k',j'), (i^r,k',j'), (i',k',j^r), (i^r,k',j^r)$ are all approximately relevant and approximately $(r-1)$-uncovered. For any such $(i',k',j')$, any $(i,k,j) \in I(i') \times I(k') \times I(j')$ gets covered in round~$r$, in fact, these are the triples counted in Lemma~\ref{lem:numuncovered} (after replacing ``weakly'' by ``approximately'' relevant and uncovered). As we maximize this number, we cover at least as many new triples as in expectation, so that Lemma~\ref{lem:numuncovered} still holds, after replacing ``weakly'' by ``approximately'' relevant and uncovered: \emph{For any $1 \le r \le \rho$ the number of approximately relevant, approximately $r$-uncovered triples is $\tilde O(n^3 / r^{1/3})$}. Since this is sufficient for the analysis of Step 3, we obtain the same running time bound as for the randomized algorithm, except that Step 2 takes additional time $O(\rho (n/\Delta)^{1+\omega})$.

\paragraph{Total running time}
Adapting the basic Algorithm~\ref{alg:basic} yields, as in Section~\ref{sec:timeanalysis}, a running time of $\tilde O(\rho \Delta W n^\omega + \rho (n/\Delta)^{1+\omega} + n^3 / \rho^{1/3} + n^{2.5})$. We optimize this by setting 
$ \Delta := ( n/W )^{1/(\omega+2)}$ and $\rho := n^{3(5+\omega-\omega^2)/(4\omega+8)} W^{-3(\omega+1)/(4\omega+8)}$.
This yields time $\tilde O(n^{3 - (5+\omega-\omega^2)/(4\omega+8)} W^{(\omega+1)/(4\omega+8)}) \le O( n^{2.9004} W^{0.1929} )$, using the current bound of $\omega \le 2.3728639$~\cite{Gall14a}. In particular, the algorithm has truly sub-cubic running time whenever $W \le O(n^{2-\omega+3/(\omega+1)- \eps}) \approx O(n^{0.5165-\eps})$ for any $\eps > 0$.

For $W = O(1)$, adapting the improved algorithm from Section~\ref{sec:improvedtimeanalysis} yields
$$ T(n) \le \tilde O\big( n^{\delta+s} \MM(n,n^{1-s/3},n) + n^{1-\delta} \MM(n^{1-\delta},n^s,n^{1-\delta}) + n^{(3-3\delta-s/3)/(1-\delta)} + n^{(1+\omega)(1-\delta) + s} \big),$$
which is $O(n^{2.8603})$ for $\delta := 0.2463$ and $s := 0.3159$, finishing the proof of Theorem~\ref{thr:mainProduct}.
Here we use the bounds $\MM(m,m^{1-s/3},m) \le \MM(m,m^{0.90},m) \le O(m^{2.298048})$ and $\MM(m,m^{s/(1-\delta)},m) \le \MM(m,m^{0.45},m) \le O(m^{2.027102})$ by Le Gall~\cite{le:12} for $m = n$ and $m = n^{1-\delta}$, respectively.

\subsection{Generalizations}

In this section we study generalizations of Theorem~\ref{thr:mainProduct}. In particular, we will see that it suffices if $A$ has bounded differences along either the columns or the rows, while $B$ may be arbitrary. Since $A\minp B = (B^T \minp A^T)^T$, a symmetric algorithm works if $A$ is arbitrary and $B$ has bounded differences along either its columns or its rows.

\begin{theorem} \label{thm:generalizations}
  Let $A,B$ be integer matrices, where $B$ is arbitrary and we assume either of the following:
  \begin{enumerate}[(1)]
    \item for an appropriately chosen $1 \le \Delta \le n$ we are given a partitioning $[n] = I_1 \cup \ldots \cup I_{n/\Delta}$ such that $\max_{i \in I_\ell} A_{i,k} - \min_{i \in I_\ell} A_{i,k} \le \Delta W$ for all $k,\ell$, or
    \item for an appropriately chosen $1 \le \Delta \le n$ we are given a partitioning $[n] = K_1 \cup \ldots \cup K_{n/\Delta}$ such that $\max_{k \in K_\ell} A_{i,k} - \min_{k \in K_\ell} A_{i,k} \le \Delta W$ for all $i,\ell$. 
  \end{enumerate}
  If $W \le O(n^{3-\omega-\eps})$, then $A\minp B$ can be computed in randomized time $O(n^{3 - \Omega(\eps)})$. If $W=O(1)$, then $A\minp B$ can be computed in randomized time $O(n^{2.9217})$.
\end{theorem}

Important special cases of the above theorem are that $A$ has $W$-BD only along columns ($|A_{i+1,k} - A_{i,k}| \le W$ for all $i,k$) or only along the rows ($|A_{i,k+1} - A_{i,k}| \le W$ for all $i,k$). In these cases the assumption is indeed satisfied, since we can choose each $I_\ell$ or $K_\ell$ as a contiguous subset of $\Delta$ elements of $[n]$, thus amounting to a total difference of at most $\Delta W$.

\begin{proof}
  (1) For the first assumption, adapting Algorithm~\ref{alg:basic} is straight-forward. Instead of blocks $I(I') \times I(k') \times I(j')$ we now consider blocks $I_\ell \times \{k\} \times \{j\}$, for any $\ell \in [n/\Delta]$, $k,j \in [n]$. Within any such block, $A_{i,k}$ varies by at most $\Delta W$ by assumption. Moreover, $B_{kj}$ does not vary at all, since $k,j$ are fixed. We adapt Step 1 by computing for each block  $I_\ell \times \{k\} \times \{j\}$ one entry $\tilde C_{i^*j} = (A\minp B)_{i^*j}$ exactly, for some $i^* \in I_\ell$, and setting $\tilde C_{ij} := \tilde C_{i^*j}$ for all other $i \in I_\ell$. It is easy to see that Lemma~\ref{lem:blockingerror} still holds. Note that Step 1 now runs in time $O(n^3 / \Delta)$.
  
  Step 2 does not have to be adapted at all, since as we remarked in Section~\ref{sec:steptwo} it works for arbitrary matrices.
  
  For Step 3, we have analogous notions of being approximately relevant or uncovered, by replacing the notion of ``blocks''. Thus, we now iterate over every $\ell, k, j$, check whether it is approximately relevant (i.e., $|A_{i^* k} + B_{kj} - \tilde C_{i^* j}| \le 8 \Delta W$ for some $i^* \in I_\ell$), check whether it is approximately uncovered (i.e., for all rounds $r$ we have $|A^r_{i^*k}| > 44 \Delta W$ or $|B^r_{kj}| > 44 \Delta W$), and if so we exhaustively search over all $i \in I_\ell$, setting $\hat C_{ij} := \min\{\hat C_{ij}, A_{ik} + B_{kj}\}$. Then Lemma~\ref{lem:relatenotions} still holds and correctness and running time analysis hold almost verbatim. Step 3 now runs in time $\tilde O(\rho n^3 / \Delta + n^3 /\rho^{1/3})$ w.h.p. 
  
  The total running time is w.h.p.\ $\tilde O(\rho \Delta W n^\omega + \rho n^3 / \Delta + n^3 /\rho^{1/3})$. We optimize this by setting $\Delta := n^{(3-\omega)/2} / W^{1/2}$ and $\rho := n^{3(3-\omega)/8} / W^{3/8}$, obtaining time $\tilde O(n^{3-(3-\omega)/8} W^{1/8})$. As desired, this is $n^{3-\Omega(\eps)}$ for $W = O(n^{3-\omega-\eps})$, while for $W=O(1)$ it evaluates to $\tilde O(n^{3-(3-\omega)/8}) \le O(n^{2.9217})$. The latter bound can be slightly improved by incorporating the improvements from Section~\ref{sec:improved}, we omit the details.

  (2') Before we consider the second assumption, we first discuss a stronger assumption where also $B$ is nice along the columns:
  Assume that for an appropriately chosen $1 \le \Delta \le n$ we are given a partitioning $[n] = K_1 \cup \ldots \cup K_{n/\Delta}$ such that $\max_{k \in K_\ell} A_{i,k} - \min_{k \in K_\ell} A_{i,k} \le \Delta W$ for all $i,\ell$ and $\max_{k \in K_\ell} B_{kj} - \min_{k \in K_\ell} B_{kj} \le \Delta W$ for all $\ell,j$. 
    
  In this case, adapting Algorithm~\ref{alg:basic} is straight-forward and similar to the last case. Instead of blocks $I_\ell \times \{k\} \times \{j\}$ we now consider blocks $\{i\} \times I_\ell \times \{j\}$, for any $\ell \in [n/\Delta]$, $i,j \in [n]$. Within any such block, $A$ and $B$ vary by at most $\Delta W$ by assumption. We adapt Step 1 by computing for each $i,\ell,j$ for some value $k^* \in K_\ell$ the sum $A_{ik^*} + B_{k^* j}$. We set $\tilde C_{ij}$ as the minimum over all $\ell$ of the computed value. It is easy to see that Lemma~\ref{lem:blockingerror} still holds. Step 1 now runs in time $O(n^3 / \Delta)$.
  
  Step 2 does not have to be adapted at all, since as we remarked in Section~\ref{sec:steptwo} it works for arbitrary matrices.
  
  For Step 3, we now iterate over every $i, \ell, j$, check whether it is approximately relevant (i.e., $|A_{i k^*} + B_{k^* j} - \tilde C_{i j}| \le 8 \Delta W$ for some $k^* \in K_\ell$), check whether it is approximately uncovered (i.e., for all rounds $r$ we have $|A^r_{ik^*}| > 44 \Delta W$ or $|B^r_{k^*j}| > 44 \Delta W$), and if so we exhaustively search over all $k \in K_\ell$, setting $\hat C_{ij} := \min\{\hat C_{ij}, A_{ik} + B_{kj}\}$. Then Lemma~\ref{lem:relatenotions} still holds and correctness and running time analysis hold almost verbatim. Step 3 now runs in time $\tilde O(\rho n^3 / \Delta + n^3 /\rho^{1/3})$ w.h.p. 
  
  We obtain the same running time as in the last case.
  
  (2) For the second assumption, compute for all $\ell,j$ the value $v(\ell,j) := \min\{B_{kj} \mid k \in K_\ell\}$, and consider a matrix $B'$ with $B'_{kj} := \min\{ B_{kj}, v(\ell,j) + 2 \Delta W\}$, where $k \in K_\ell$. Note that for any $i,k,j$ with $k \in K_\ell$ and $k^* \in K_\ell$ such that $B_{k^*j} = v(\ell,j)$, we have $A_{ik} + (v(\ell,j) + 2 \Delta W) \ge A_{i k^*} + B_{k^* j} + \Delta W > C_{ij}$, since $A$ varies by at most $\Delta W$. Hence, no entry $B_{kj} = v(\ell,j) + 2 \Delta W$ is strongly relevant, which implies $A\minp B' = A\minp B$. Note that $B'$ satisfies $\max_{k \in K_\ell} B_{kj} - \min_{k \in K_\ell} B_{kj} \le 2 \Delta W$ for all $\ell,j$, so we can use case (2') to compute $A\minp B'$.
  Since $B'$ can be computed in time $O(n^2)$, the result follows.
\end{proof}

\section{Fast Scored Parsing}
\label{sec:scoredParsing}

In this section, we prove Theorem \ref{thr:SP} that reduces the scored parsing problem for BD grammars to the $(\min,+)$-product for BD matrices. For a square matrix $M$, we let $n(M)$ denote its number of rows and columns.

We will exploit a generalization of Valiant's parser \cite{v97}. We start by describing Valiant's classic approach in Section \ref{sec:valiant}. Then in Section \ref{sec:scoredValiant} we show how to modify Valiant's parser to solve the scored parsing problem, thereby replacing Boolean matrix multiplications by $(\min,+)$-products.
Finally, in Section \ref{sec:valiantAnalysis} we show that all $(\min,+)$-products in our scored parser involve BD matrices.

\subsection{Valiant's Parser}
\label{sec:valiant}

Given a CFG $G=(N,T,P,S)$ and a string $\sigma=\sigma_1\sigma_2....\sigma_n\in T^*$, the \emph{parsing} problem is to determine whether $\sigma\in L(G)$. In a breakthrough paper \cite{v97}, Valiant presented a reduction from parsing to Boolean matrix multiplication, which we describe in the following (for a more detailed description, see \cite{v97}). Let us define a (product) operator ``.'' as follows. For $N_1,N_2\subseteq N$,
$$
N_1.N_2 = \{Z\in N : \exists X\in N_1, \exists Y\in N_2: (Z\rightarrow XY)\in P\}.
$$
Note the above operator is \emph{not associative} in general, namely $(N_1.N_2).N_3$ might be different from $N_1.(N_2.N_3)$. 

Given a $a\times b$ matrix $A$ and a $b\times c$ matrix $B$, whose entries are subsets of $N$, we can naturally define a matrix product $C = A . B$, where
$C_{i,j} = \bigcup_{k=1}^{b} A_{i,k}.B_{k,j}$. 
Observe that this ``.'' operator can be reduced to the computation of a constant\footnote{Here we ignore the (polynomial) dependence on the size of the grammar $G$, as we assume for simplicity that $G$ has constant size.} number of standard Boolean matrix multiplications. Indeed, for a matrix $M$ and non-terminal $X$, we let $M(X)$ be the $0$-$1$ matrix with the same dimensions as $X$ and entries $M(X)_{i,j}=1$ iff $X\in M_{i,j}$. In order to compute the product $C=A.B$, we initialize matrix $C$ with empty entries. Then we consider each production rule $Z\rightarrow XY$ separately, and we compute $C'(Z)=A(X)\cdot B(Y)$, where $\cdot$ is the standard Boolean matrix multiplication. Then, for all $i,j$, we add $Z$ to the set $C_{i,j}$ if $C'(Z)_{i,j}=1$.

The transitive closure $A^+$ of an $m\times m$ matrix $A$ of the above kind is defined as 
$$
A^+ = \bigcup_{i=1}^{m}A^{(i)},
$$
where 
$$
A^{(1)}=A\quad \text{and}\quad
A^{(i)} = \bigcup_{j=1}^{i-1}A^{(j)}.A^{(i-j)}.
$$ 
Here unions are taken component-wise.

Given the above definitions we can formulate the parsing problem as follows. We initialize an $(n+1)\times (n+1)$ matrix $A$ with $A_{i,i+1}= \{X\in N : (X\rightarrow \sigma_i)\in P\}$ and $A_{i,j}=\emptyset$ for $j\neq i+1$.  Then by the definition of the operator ``.'' it turns out that $X\in (A^+)_{i,j}$ if and only if $\sigma_i\ldots \sigma_{j-1}\in L(X)$. Hence one can solve the parsing problem by computing $A^+$ and checking whether $S\in A^+_{1,n+1}$.

Suppose that, for two given $n\times n$ matrices, the ``.'' operation can be performed in time $O(n^{\alpha})$ for some $2\leq \alpha \leq 3$, and note that the ``$\cup$'' operation can be performed in time $O(n^2)$. Crucially, we cannot simply use the usual squaring technique to compute $A^+$ in time $\tO(n^{\alpha})$, due to the fact that ``.'' is not associative. However, Valiant describes a more sophisticated approach to achieve the same running time. It then follows that the parsing problem can be solved in time $\tilde{O}(n^\omega)$, where $2\leq \omega<2.373$ is the exponent of fast Boolean matrix multiplication \cite{v12,Gall14a} ($O(n^\omega)$ if $\omega>2$).  

\begin{algorithm}
\caption{Valiant's parser. In all the subroutines the input is an $n(B)\times n(B)$ matrix $B$, which is passed by reference. 
By $B_{I}^{J}$ we denote the submatrix of $B$ having entries $B_{i,j}$, with $i\in I$ and $j\in J$.
}
\label{alg:valiant}
Parse(B):
%{\small
\begin{algorithmic}[1]
\If {$n(B)> 1$}
\State Parse($B_{[1,n(B)/2]}^{[1,n(B)/2]}$)
\State Parse($B_{[n(B)/2+1,n(B)]}^{[n(B)/2+1,n(B)]}$)
\State Parse$_2$($B$)
\EndIf
\end{algorithmic}
%}
Parse$_2$($B$):
%{\small
\begin{algorithmic}[1]
\If {$n(B)> 2$}
\State Parse$_2$($B_{[n(B)/4+1,3n(B)/4]}^{[n(B)/4+1,3n(B)/4]}$)
\State Parse$_3$($B_{[1,3n(B)/4]}^{[1,3n(B)/4]}$)
\State Parse$_3$($B_{[n(B)/4+1,n(B)]}^{[n(B)/4+1,n(B)]}$)
\State Parse$_4$($B$)
\EndIf
\end{algorithmic}
%}
Parse$_3$($B$):
\begin{algorithmic}[1]
\State Let $I_1=[1,n(B)/3]$, $I_2=[n(B)/3+1,2n(B)/3]$, and $I_3=[2n(B)/3+1,n(B)]$
\State $B_{I_1}^{I_3} \leftarrow B_{I_1}^{I_3} \cup (B_{I_1}^{I_2}\,.\,B_{I_2}^{I_3})$
\State $C\leftarrow$ matrix obtained from $B$ by deleting  row/column indices in $I_2$
\State Parse$_2$($C$)
\State $B\leftarrow $ matrix obtained from $C$ by reintroducing the rows and columns deleted in Step 3
\end{algorithmic}
Parse$_4$($B$):
\begin{algorithmic}[1]
\State Let $I_1=[1,n(B)/4]$, $I_2=[n(B)/4+1,2n(B)/4]$, $I_3=[2n(B)/4+1,3n(B)/4]$, and $I_4=[3n(B)/4+1,n(B)]$
\State $B_{I_1}^{I_4} \leftarrow B_{I_1}^{I_4} \cup ( B_{I_1}^{I_2}\,.\,B_{I_2}^{I_4}) \cup (B_{I_1}^{I_3}\,.\,B_{I_3}^{I_4})$
\State $C\leftarrow$ matrix obtained from $B$ by deleting  row/column indices in $I_2 \cup I_3$
\State Parse$_2$($C$)
\State $B\leftarrow $ matrix obtained from $C$ by reintroducing the rows and columns deleted in Step 3
\end{algorithmic}
\end{algorithm}

For the sake of simplicity we assume that $n+1$ is a power of $2$. This way we can avoid the use of ceilings and floors in the definition of some indices. It is not hard to handle the general case either by introducing ceilings and floors or by introducing dummy entries (with a mild adaptation of some definitions).
Valiant's fast procedure to compute the transitive closure of a given matrix is described in Algorithm \ref{alg:valiant}.
In this algorithm, we use the following notation:
For two sets of indices $I$ and $J$, by $B_{I}^{J}$ we denote the submatrix of $B$ given by entries $B_{i,j}$, with $i\in I$ and $j\in J$.
The algorithm involves $4$ recursive procedures: Parse, Parse$_2$, Parse$_3$, and Parse$_4$. Each one of them receives as input an $n(B)\times n(B)$ matrix $B$, and the result of the computation is stored in $B$ (i.e., $B$ is passed by reference). 
Assuming that $n+1$ is a power of $2$, it is easy to see that the sizes of the input matrices to Parse and Parse$_2$ are all powers of $2$. This guarantees that all the indices used in the algorithm are integers (this way we can avoid ceilings and floors as mentioned earlier).

The running time bound $\tilde O(n^{\omega})$ follows by standard arguments.
For the (subtle) correctness argument we refer to \cite{v97}, but the argument is also implicit in Lemma~\ref{lem:alreadyClosed} below.

\subsection{Scored Parser and \boldmath$(\min,+)$-Products}
\label{sec:scoredValiant}

We can adapt Valiant's approach to scored parsing as follows.\footnote{This has already been done in~\cite{s:15}, but we give details here for the sake of completeness.}
Let $G = (N,T,P,S)$ be a scored grammar with score function $s$ (mapping productions to non-negative integers).
Let us consider the set $\mathcal{F}_N$ of all functions $F\colon N \to \mathbb{N}_{\ge 0} \cup \{\infty\}$. We interpret $F(X)$ as the score of non-terminal $X \in N$, and thus $F$ is a score function on the non-terminals. We write $\overline{\infty}$ for the score function mapping each $X \in N$ to $\infty$.
Let $F_1,F_2 \in {\mathcal F}_N$. We redefine the operator~``$\cup$'' as pointwise minimum: 
$$(F_1 \cup F_2)(X) := \min\{F_1(X), F_2(X)\}.$$ 
We also redefine the operator~``.'' as follows (where the minimum is $\infty$ if the set is empty):
$$  (F_1.F_2)(X) = \min\{ s(X \to YZ) + F_1(Y) + F_2(Z) \mid (X \to YZ) \in P\}. $$

Given the above operations ``.'' and ``$\cup$'', we can define the product of two matrices whose entries are in $\mathcal{F}_N$ as well as the transitive closure of one such square matrix in the same way as before, i.e., $C = A.B$ is defined via $C_{i,j} = \bigcup_{k} A_{i,k}.B_{k,j}$, and for an $m\times m$-matrix $A$ we have $A^+ = \bigcup_{i=1}^m A^{(i)}$ where $A^{(1)} = A$ and $A^{(i)} = \bigcup_{j=1}^{i-1} A^{(j)}.A^{(i-j)}$. 
We can then solve the scored parsing problem as follows. For a given string $\sigma$ of length~$n$, we define a $(n+1)\times (n+1)$ matrix~$A$ whose entries are in $\mathcal{F}_N$, where 
$$A_{i,i+1}(X)= \min\{ s(X\rightarrow \sigma_i) \mid (X\rightarrow \sigma_i) \in P \}, $$ 
for $i=1,\ldots, n$ and $A_{i,j}=\overline{\infty}$ for $j\neq i+1$. 
Then by the definition of the operator ``.'' it follows that $(A^+)_{i,j}$ evaluated at $X$ equals the score $s(X, \sigma_i\ldots \sigma_{j-1})$. Hence, the solution to the scored parsing problem is $(A^+)_{1,n+1}$ evaluated at the starting symbol $S \in N$.

Crucially for our goals, the ``.'' operator can be implemented with a reduction to a constant (for constant grammar size) number of $(\min,+)$-products $\minp$, with a natural adaptation of the previously described reduction to Boolean matrix multiplication. For a matrix $M$ with entries in $\mathcal{F}_N$ and for $X\in N$, let $M(X)$ be the matrix with the same dimension as $M$ and having $M(X)_{i,j}=M_{i,j}(X)$. With the same notation as before, in order to compute the product $C=A.B$ we initialize matrix $C$ with $\overline{\infty}$ entries. Then we consider each production rule $Z\rightarrow XY$ separately, and we compute $C'(Z)=A(X) \minp B(Y)$. Then we set $C_{i,j}(Z) = \min\{C_{i,j}(Z), s(Z \to XY) + C'(Z)_{i,j}\}$ for all $i,j$. This computes $C = A.B$.

With the above modifications, the same Algorithm~\ref{alg:valiant} computes $A^+$ in the scored setting. This is proven formally in the next section.

\subsection{Reduction to Bounded-Difference \boldmath$(\min,+)$-Product}
\label{sec:valiantAnalysis}

In this section, we show that Algorithm~\ref{alg:valiant} also works in the scored setting. More importantly, we prove that the matrix products $B_I^J . B_{I'}^{J'}$ called by this algorithm can be implemented using $(\min,+)$-products of $W$-BD matrices, if the scored grammar is $W$-BD (recall Definition~\ref{def:grammarBD}). This allows us to use our main result to obtain a good running time bound for Algorithm~\ref{alg:valiant} and thus for scored parsing of BD grammars.

We start by proving the correctness of Algorithm~\ref{alg:valiant} in the scored setting, see Lemma~\ref{lem:alreadyClosed} below. Some properties that we show along the way will be also crucial for the BD property and running time analysis.  

We first prove a technical lemma that relates the indices of the input square matrices $B$ in the various procedures to the indices of the original matrix $A$. Note that each such matrix $B$ corresponds to some submatrix of $A$, however indices of $B$ might map discontinuously to indices of $A$ (i.e., the latter indices do not form one interval). This is due to Step 3 of Parse$_3$($B$) and Parse$_4$($B$) that constructs a matrix $C$ by removing central rows and columns of $B$. Note also that by construction the row indices of $A$ associated to $B$ are equal to the corresponding column indices (since the mentioned step removes the same set of rows and columns). We denote by $\map{B}{i}$ the row/column index of $A$ corresponding to row/column index $i$ of $B$. We say that $B$ is \emph{contiguous} if $\{\map{B}{i}\}_{i=1,\ldots,n(B)}=\{\map{B}{1},\map{B}{1}+1,\ldots, \map{B}{1}+n(B)-1\}$. In other words, the indices of $A$ corresponding to $B$ form an interval of contiguous indices. We say that $B$ has a \emph{discontinuity} at index $1< a< n(B)$ if $B$ is not contiguous but the submatrices $B_{[1,a]}^{[1,a]}$ and $B_{[a+1,n(B)]}^{[a+1,n(B)]}$ are contiguous. We call the indices $J=\{\map{B}{a}+1,\ldots,\map{B}{a+1}-1\}$ the \emph{missing indices} of $B$.

\begin{lemma}\label{lem:discontinuity}
Any input matrix $B$ considered by the procedures in the scored parser:
\begin{enumerate}\itemsep0pt
\item is contiguous if it is the input to Parse;
\item is contiguous or has a discontinuity at $n(B)/2$ if it is the input to Parse$_2$ or Parse$_4$;
\item is contiguous or has a discontinuity at $n(B)/3$ or $2n(B)/3$ if it is the input to Parse$_3$. 
%\item is contiguous or has a discontinuity at $n(B)/2$ if it is the input to Parse$_2$;
\end{enumerate}
\end{lemma}
\begin{proof}
We prove the claims by induction on the partial order induced by the recursion tree.

Parse($B$) satisfies the claim in the starting call with $B=A$. In the remaining cases Parse($B$) is called by Parse($D$) with $B=D_{[1,n(D)/2]}^{[1,n(D)/2]}$ or $B=D_{[n(D)/2+1,n(D)]}^{[n(D)/2+1,n(D)]}$. The claim follows by inductive hypothesis on Parse($D$).

Parse$_4$($B$) is called by Parse$_2$($B$). The claim follows by inductive hypothesis on Parse$_2$($B$).

Parse$_3$($B$) is called by Parse$_2$($D$), with (i) $B=D_{[1,3n(D)/4]}^{[1,3n(D)/4]}$ or (ii) $B=D_{[n(D)/4+1,n(D)]}^{[n(D)/4+1,n(D)]}$. By inductive hypothesis $D$ is contiguous or has a discontinuity at $n(D)/2$. Hence $B$, if not contiguous, has a discontinuity at $2n(B)/3$ in case (i) and at $n(B)/3$ in case (ii). The claim follows.

Finally consider Parse$_2$($B$). If it is called by Parse($B$), the claim follows by inductive hypothesis on Parse($B$). If it is called by Parse$_2$($D$) with $B=D_{[n(D)/4+1,3n(D)/4]}^{[n(D)/4+1,3n(D)/4]}$, the claim follows by inductive hypothesis on Parse$_2$($D$). Suppose it is called by Parse$_4$($D$). Then $D$ has size $n(D)=2n(B)$, and is contiguous or has a discontinuity at $n(D)/2$ by inductive hypothesis. In this case $B$ is obtained by removing the $n(D)/2$ central columns and rows of $D$. Therefore $B$ has a discontinuity at $n(B)/2$. The remaining case is that Parse$_2$($B$) is called by Parse$_3$($D$), where $D$ has size $n(D)=3n(B)/2$. In this case $B$ is obtained by removing the $n(D)/3$ central columns and rows of $D$. Since $D$ is contiguous or has a discontinuity at $n(D)/3$ or $2n(D)/3$ by inductive hypothesis on Parse$_3$($D$), $B$ has a discontinuity at $n(B)/2$. 
\end{proof}

The following lemma proves the correctness of our algorithm, and is also crucial to analyse its running time.

\begin{lemma}\label{lem:alreadyClosed}
Let $A$ be the input matrix and $B$ be any submatrix in input to some call to \textup{Parse}$_k$, $k\in \{2,3,4\}$. Then we have the \emph{input property}
\begin{align}
& B_{i,j} = (A^+)_{\map{B}{i},\map{B}{j}}\quad \forall i,j\in [1,n(B)-n(B)/k], \nonumber\\
& B_{i,j} = (A^+)_{\map{B}{i},\map{B}{j}}\quad \forall i,j\in [n(B)/k+1,n(B)]. \label{eqn:alreadyClosed1}
\end{align}
Furthermore, let $J$ be the missing indices in $B$ if $B$ has a discontinuity, and let $J=\emptyset$ if $B$ is contiguous. Then for all $i\in [1,n(B)/k]$ and $j\in [n(B)-n(B)/k+1,n(B)]$ we have the additional \emph{input property}
\begin{align}
B_{i,j} = A_{\map{B}{i},\map{B}{j}} \cup \bigcup_{k\in J}(A^+)_{\map{B}{i},k}.(A^+)_{k,\map{B}{j}}. \label{eqn:alreadyClosed3}
\end{align}
The matrix $B$ at the end of the procedure has the following \emph{output property}
\begin{align*}
B_{i,j} = (A^+)_{\map{B}{i},\map{B}{j}}\quad \forall i,j\in [1,n(B)]. %\label{eqn:alreadyClosed4}
\end{align*}
The same output property holds for procedure \textup{Parse}. 
\end{lemma}
\begin{proof}
We prove the claim by induction on the total order defined by the beginning and the end of each procedure during the execution of the algorithm starting from Parse($A$). 

Consider the input property of some call Parse$_2$($B$).  Suppose Parse$_2$($B$) is called in Step 4 of Parse($B$). Let $I_1=[1,n(B)/2]$ and $I_2=[n(B)/2+1,n(B)]$. 
The input property \eqref{eqn:alreadyClosed1} follows by the output property of Parse$_2$($B_{I_1}^{I_1}$) and Parse$_2$($B_{I_2}^{I_2}$) called in Steps 2 and 3. Since $B$ is contiguous, we have $J = \emptyset$, and thus input property \eqref{eqn:alreadyClosed3} requires $B_{i,j} = A_{\map{B}{i},\map{B}{j}}$ for all $i \in I_1$ and $j \in I_2$. 
Since in the calls of Parse($B$) the input top-right quadrant is as in the initial matrix $A$, and this is not affected by Steps 2 and 3, input property \eqref{eqn:alreadyClosed3} is satisfied.

If Parse$_2$($B$) is called in Step 2 of Parse$_2$($D$) for some $D$, the input property follows by inductive hypothesis on the input property of Parse$_2$($D$). Otherwise Parse$_2$($B$) is called in Step~4 of Parse$_k$($D$), for some $D$ and $k\in \{3,4\}$. Input property \eqref{eqn:alreadyClosed1} directly follows from the input property of Parse$_k$($D$). It remains to show that input property \eqref{eqn:alreadyClosed3} holds. Let $S=[1,n(D)/k]$, $M=[n(D)/k+1,n(D)-n(D)/k]$, and $L=[n(D)-n(D)/k+1,n(D)]$. We need to show that $B_S^L$ has the desired property. 
%If $D$ is contiguous, the input property of $D$ implies that $B_{i,j} = A_{\map{B}{i},\map{B}{j}}$ for all $i \in S$ and $j \in L$.
%Thus Step 2 of Parse$_k$($D$) enforces precisely the input property \eqref{eqn:alreadyClosed3} on $B$. 
Let $J$ be the missing indices in $D$ (or $\emptyset$, if $D$ is contiguous). Observe that, by Step~3 of Parse$_k$($D$), the missing indices in $B$ will be $J'=M\cup J$. By the input property \eqref{eqn:alreadyClosed3} of $D$, we have  
$$
D_{i,j} = A_{\map{D}{i},\map{D}{j}}  \cup \bigcup_{k\in J}(A^+)_{\map{D}{i},k}.(A^+)_{k,\map{D}{j}}\quad \forall i\in S,\forall j\in L.
$$
Therefore at the end of Step 2 of Parse$_k$($D$) one has, for all $i\in S$ and $j\in L$,
\begin{align*}
D_{i,j}  & = A_{\map{D}{i},\map{D}{j}}  \cup \Big( \bigcup_{k\in J}(A^+)_{\map{D}{i},k}.(A^+)_{k,\map{D}{j}} \Big) \cup \bigcup_{k\in M}D_{i,k}.D_{k,j}\\
& = A_{\map{D}{i},\map{D}{j}}  \cup \Big( \bigcup_{k\in J}(A^+)_{\map{D}{i},k}.(A^+)_{k,\map{D}{j}} \Big) \cup \bigcup_{k\in M} (A^+)_{\map{D}{i},k}.(A^+)_{k,\map{D}{j}}\\
& = A_{\map{D}{i},\map{D}{j}}  \cup \bigcup_{k\in J \cup M}(A^+)_{\map{D}{i},k}.(A^+)_{k,\map{D}{j}} ,
\end{align*}
where in the second equality we used the input property \eqref{eqn:alreadyClosed1} of Parse$_k$($D$). This implies input property \eqref{eqn:alreadyClosed3} for Parse$_2$($B$).

Let us consider the input property of some call Parse$_3$($B$). Note that Parse$_3$($B$) is called either in Step 3 or in Step 4 of Parse$_2$($D$) for some $D$. Consider the first case, the second one is analogous. Let $I_1=[1,n(D)/4]$, $I_2=[n(D)/4+1,2n(D)/4]$, $I_3=[2n(D)/4+1,3n(D)/4]$, and $I_4=[3n(D)/4+1,n(D)]$. By the input property of $D$ and the output property of Parse$_2$($D_{I_2\cup I_3}^{I_2\cup I_3}$), the input property \eqref{eqn:alreadyClosed1} of $B$ follows. The input property \eqref{eqn:alreadyClosed3} of $B$ follows directly from the input property \eqref{eqn:alreadyClosed3} of $D$ since the missing indices in $D$ and $B$ are the same.

Finally consider the input property of Parse$_4$($B$). Note that Parse$_4$($B$) is called in Step 5 of Parse$_2$($B$). With the same notation as above, the input property \eqref{eqn:alreadyClosed1} of $B$ follows directly from the output property of Parse$_3$($B_{I_1\cup I_2\cup I_3}^{I_1\cup I_2\cup I_3}$) and of Parse$_3$($B_{I_2\cup I_3\cup I_4}^{I_2\cup I_3\cup I_4}$). The input property \eqref{eqn:alreadyClosed3} of $B$ follows directly from the input property \eqref{eqn:alreadyClosed3} of $B$ at the beginning of Parse$_2$($B$) since $B_{I_1}^{I_4}$ is not modified by Steps 2-4. 

It remains to discuss the output properties. Let us start with the output property of Parse($B$). The base case is $n(B)=1$. In this case the unique entry $B_{1,1}$ corresponds to an entry in the main diagonal of the input matrix, and these entries are never updated by the algorithm. In other words, $B_{1,1}=A_{\map{B}{1},\map{B}{1}}$ where $A$ is the input matrix (in particular, this is a trivial entry with all $\infty$ values). This is the correct answer since trivially $(A^+)_{i,i} = A_{i,i} = \overline{\infty}$ for all $i=1,\ldots,n+1$.
For $n(B)\geq 2$, the output property follows from the output property of Parse$_2$($B$). 

Consider next the output property of Parse$_2$($B$). The base case is $n(B)=2$. In this case the input property of $B$ coincides with its output property. More precisely, let $J$ be the indices strictly between $i=\map{B}{1}$ and $j=\map{B}{2}$. Then the entry $B_{1,2}$ satisfies
$$
B_{1,2} = A_{i,j} \cup \bigcup_{k\in J}(A^+)_{i,k}.(A^+)_{k,j}=(A^+)_{i,j},
$$
where the last equality follows from the definition of $A^+$.
For $n(B)>2$, the output property follows from the output property of Parse$_4$($B$). 

Consider the output property of Parse$_3$($B$). Let $I_1=[1,n(B)/3]$, $I_2=[n(B)/3+1,2n(B)/3]$, and $I_3=[2n(B)/3+1,n(B)]$. By the input property \eqref{eqn:alreadyClosed1} of $B$, at the beginning of the procedure the only part of $B$ which might not satisfy the output property is $B_{I_1}^{I_3}$. This property is enforced on $B_{I_1}^{I_3}$ at the end of Step 4 due to the output property of Parse$_2$($C$). The output property of Parse$_4$($B$) can be shown analogously: Here the part of $B$ that needs to fixed is $B_{I_1}^{I_4}$, with $I_1=[1,n(B)/4]$ and $I_4=[3n(B)/4+1,n(B)]$. This is done in Step 4 due to the output property of Parse$_2$($C$).
\end{proof}

It remains to argue that all (explicit) matrix products performed by the scored parser can be implemented using $(\min,+)$-products of BD matrices. 

Recall that in the scored parser the only explicit matrix products that we perform are of type $B_I^J . B_J^K$ in procedures Parse$_k$($B$), $k\in \{3,4\}$.
Recall that in order to implement a product $B_I^J.B_J^K$ we consider each production rule $Z\rightarrow XY$ (the number of such rules is constant), we derive integer matrices $B_I^J(X)$ and $B_J^K(Y)$, and then we compute the $(\min,+)$-product $B_I^J(X)\minp B_J^K(Y)$. In the next corollary we prove that each such product involves two BD matrices. Therefore we can perform it in time $O(n(B)^\alpha)$ for some $2 \le \alpha < 3$ using our faster algorithm for BD $(\min,+)$-product. It follows from the previous discussion that the overall running time of our scored parser is $\tO(n^\alpha)$ (or $O(n^\alpha)$ if $\alpha > 2$).

\begin{lemma}\label{cor::alreadyClosed}
If the scored grammar $G$ is $W$-BD, then the products in Step 2 of Parse$_k$, $k\in \{3,4\}$, involve $W$-BD submatrices. 
\end{lemma}
\begin{proof}
Consider first Parse$_3$($B$). Recall that we perform the product $B_{I_1}^{I_2}\,.\,B_{I_2}^{I_3}$, where $I_1=[1,n(B)/3]$, $I_2=[n(B)/3+1,2n(B)/3]$  and $I_3=[2n(B)/3+1,n(B)]$. By Lemma \ref{lem:discontinuity}, $B$ is contiguous or has a discontinuity at $n(B)/3$ or $2n(B)/3$. Thus each such $I_j$ forms a contiguous set of indices w.r.t.\ the input matrix $A$. By Lemma \ref{lem:alreadyClosed} (input property), the submatrices $B_{I_1}^{I_2}$ and $B_{I_2}^{I_3}$ are equal to the corresponding contiguous submatrices of $A^+$. Since the scored grammar is $W$-BD, and since $(A^+)_{i,j}$ evaluated at non-terminal $X$ equals the score $s(X,\sigma_i \ldots \sigma_{j-1})$, the matrix $A^+(X)$ is $W$-BD for any $X \in N$.
It follows that also $B_{I_1}^{I_2}(X)$ and $B_{I_2}^{I_3}(X)$ are $W$-BD for any $X \in N$. 

The proof in the case of Parse$_4$($B$) is analogous. Recall that we perform the products $B_{I_1}^{I_2}\,.\,B_{I_2}^{I_4}$ and $B_{I_1}^{I_3}\,.\,B_{I_3}^{I_4}$, where $I_1=[1,n(B)/4]$, $I_2=[n(B)/4+1,2n(B)/4]$, $I_3=[2n(B)/4+1,3n(B)/4]$, and $I_4=[3n(B)/4+1,n(B)]$. By Lemma \ref{lem:discontinuity}, $B$ is contiguous or has a discontinuity at $n(B)/2$. Thus each such $I_j$ forms a contiguous set of indices w.r.t.\ the input matrix $A$. By Lemma \ref{lem:alreadyClosed} (input property), the submatrices $B_{I_1}^{I_2}$,  $B_{I_1}^{I_3}$, $B_{I_2}^{I_4}$, and $B_{I_3}^{I_4}$ are equal to contiguous submatrices of $A^+$. Since $A^+(X)$ is $W$-BD for any $X \in N$, also $B_{I_1}^{I_2}(X)$,  $B_{I_1}^{I_3}(X)$, $B_{I_2}^{I_4}(X)$, and $B_{I_3}^{I_4}(X)$ are $W$-BD for any $X \in N$.
\end{proof}

\section{Applications}
\label{sec:appl}

We show that LED, RNA-folding, and OSG can be cast as scored parsing problems on BD grammars. To apply Theorem~\ref{thr:SP} we also have to make sure that the grammars are in CNF. To relax the latter condition, we first show that it suffices to obtain grammars that are ``almost CNF'', as is made precise in the following section.

Recall that a scored grammar $G$ is $W$-BD if for any non-terminal $X$, terminal $x$, and string of terminals $\sigma \ne \eps$ the following holds: $$\big| s(X,\sigma)-s(X,\sigma x)\big| \leq W ~~~~\text{ and }~~~~ \big| s(X,\sigma)-s(X,x\sigma)\big| \leq W.$$

\subsection{From Almost-CNF to CNF}
\label{sec:almostcnf}

\begin{definition}
  We call a (scored) grammar $G$ \emph{almost-CNF} if every production is of the form
  \begin{itemize}
    \item $X \to Y Z$ for non-terminals $X,Y,Z$,
    \item $X \to c$ for a non-terminal $X$ and a terminal $c$,
    \item $X \to \eps$ for a non-terminal $X$, or 
    \item $X \to Y$ for non-terminals $X,Y$.
  \end{itemize}
  That is, we relax CNF by allowing (1) $\eps$-productions for all non-terminals, (2) unit productions $X \to Y$, and (3) the starting symbol to appear on the right-hand-side.
\end{definition}

We show that any scored grammar that is almost-CNF can be transformed into a scored grammar in CNF, keeping BD properties. Hence, for our applications it suffices to design almost-CNF grammars.

\begin{lemma} \label{lem:almostcnf}
  Let $G$ be a scored grammar $G$ that is almost-CNF. In time $O(\textup{poly}(|G|))$ we can compute a scored grammar $G'$ in CNF generating the same scored language as $G$, i.e., for the start symbols $S,S'$ of $G,G'$, respectively, and any string of terminals $\sigma \ne \eps$ we have $s_G(S,\sigma) = s_{G'}(S',\sigma)$. Moreover, if $G$ is $W$-BD then $G'$ is also $W$-BD.
\end{lemma} 

The remainder of this section is devoted to the proof of this lemma. We follow the standard conversion of context-free grammars into CNF, but we can skip some steps since $G$ is already almost-CNF. In this conversion, whenever we add a new production $X \to \alpha$ with score $s$, if there already exists the production $X \to \alpha$ with score $s'$, then we only keep the production with lowest cost $\min\{s,s'\}$. We denote the set of non-terminals of $G$ by $N$ and the set of productions by $P$. The size $|G|$ is equal to $|N| + |P|$ up to constant factors.

\paragraph{Eliminating $\eps$-Productions} We eliminate productions of the form $X \to \eps$ as follows.

(Step 1) For any non-terminal $X$ from which we can derive the empty string $\eps$, let $s$ be the lowest score of any derivation $X \to^* \eps$. We add the production $X \to \eps$ with score $s$. %Note that this does not change the scored language produced by $X$, since we only added a shortcut. 

(Step 2) For any production $p$ of the form $X \rightarrow Y Z$ or $X \to Z Y$ where $Y \rightarrow \varepsilon$ is a production in the current grammar and $X \neq Z$, add a new production $X \rightarrow Z$ with a score of $s(X\rightarrow Z)=s(p) + s(Y \to \varepsilon)$. 

(Step 3) Delete all productions of the form $X \to \eps$. 

Note that this does not change the set of non-terminals.\footnote{Some non-terminals might not appear in any productions anymore; we still keep them in the set of non-terminals $N$.} 
Call the resulting grammar $G_1$.
We claim that any non-terminal generates the same scored language in $G$ and $G_1$, except that we delete the empty string from this language. In particular, the BD property is not affected, as it ignores the empty string.
To prove the claim, consider any derivation $X \to^* \sigma$ in $G$, where $\sigma \ne \eps$ is a string of terminals. Consider any non-terminal $Y$ that appears in the derivation and generates the empty string $\eps$, such that $Y$ was derived from a production $p$ of the form $A \to B Y \mid Y B$. Then we can replace the use of $p$ by the newly added production $A \to B$, while not increasing the score. Iterating this eventually yields a derivation not using any $\eps$-production, i.e., a derivation in $G_1$. For the other direction, by construction we can replace any newly added production in $G_1$ by a derivation in $G$ with the same score.

\paragraph{Efficient Implementation}
Steps 2 and 3 clearly run in linear time. To efficiently implement Step 1, we use a Bellman-Ford-like algorithm. For each non-terminal $X$ initialize $s_X$ as the score of the production $X \to \eps$, or as $\infty$, if no such production exists. At the end of the algorithm, $s_X$ will hold the minimal cost of any derivation $X \to^* \eps$, or $\infty$ if there is no such derivation. Repeat the following for $|N|$ rounds. For each production of the form $X \to YZ$ with score $s$, set $s_X := \min\{s_X, s + s_Y + s_Z\}$. For each production of the form $X \to Y$ with score $s$, set $s_X := \min\{s_X, s + s_Y\}$. 

The running time of this algorithm is clearly $O(|N| \cdot |P|) \le O(|G|^2)$. Correctness is implied by the following claim, asserting that we can restrict our attention to derivations of depth at most $|N|$, where the depth of a derivation is to be understood as the depth of the corresponding parse tree. Observe that all such derivations are incorporated in the output of our algorithm.

\begin{claim}
  For any non-terminal $X$ such that there exists a derivation $X \to^* \eps$, let $s$ be the minimal score of any such derivation. Then there exists a derivation $X \to^* \eps$ of score $s$ and depth at most $|N|$.
\end{claim}
\begin{proof}
Among the derivations $X \to^* \eps$ of (minimal) score $s$, consider one of minimum length. In this derivation, the non-terminal $X$ cannot appear anywhere (except for the first step), as any appearance of $X$ would give rise to another derivation $X \to^* \eps$, with score at most $s$ and smaller length, which is a contradiction.

We can now argue inductively. If the first production is $X \to YZ$, then the remaining derivations $Y \to^* \eps$ and $Z \to^* \eps$ without loss of generality only use non-terminals in $N \setminus \{X\}$, and thus inductively they have depth at most $|N|-1$. This yields depth at most $|N|$ for the derivation $X \to^* \eps$.
\end{proof}

\paragraph{Eliminating the Start Symbol from the Right-Hand-Side}
Let $S$ be the start symbol of the grammar $G_1$ resulting from the last step. We introduce a new non-terminal $S'$ and add the production $S' \to S$, making $S'$ the new start symbol. This does not change the generated language and eliminates the start symbol from the right-hand-side of all productions. Moreover, since $S'$ generates the same language as $S$, it inherits the BD property, so the resulting grammar has the same BD properties as $G$.

If the original grammar $G$ can generate the empty string, then we add the production $S' \to \eps$. Since $G_1$ generates the same language as $G$ except that we delete the empty string, the resulting grammar $G_2$ generates exactly the same language as $G$. Moreover, since the BD property ignores the empty string, it is not affected by this change. 

\paragraph{Eliminating Unit Productions}
We now eliminate productions of the form $X \to Y$. Interpret any production $X \to Y$ with score $s$ as an edge from vertex $X$ to vertex $Y$ with weight $s$, and compute all-pairs-shortest-paths on the resulting graph. Using Dijkstra, this runs in time $\tO(|N| \cdot |P|) \le \tO(|G|^2)$. Iterate over all productions $X \to \alpha$ with $\alpha$ of the form $YZ$ (for non-terminals $Y,Z$) or of the form $c$ (for a terminal $c$). Iterate over all non-terminals $W$, and let $s$ be the shortest path length from $W$ to $X$. If $s < \infty$, add the production $W \to \alpha$ with score $s(W \to \alpha) = s + s(X \to \alpha)$. Finally, delete all productions of the form $X \to Y$. 

It is easy to see that this procedure for eliminating unit productions runs in time $\tO(|N| \cdot |P|) \le \tO(|G|^2)$ (note that up to the construction of $G_1$ we increased the sizes of $N$ and $P$ at most by constant factors).
We claim that in the resulting grammar $G'$, any non-terminal generates the same scored language as in $G_2$. Hence, BD properties are again not affected by this change. To prove the claim, consider any derivation $X \to^* \sigma$ in $G_2$, where $\sigma \ne \eps$ is a string of terminals. In this derivation, replace any maximal sequence of unit productions followed by a non-unit production by the corresponding newly added production in $G'$. This yields a derivation $X \to \sigma$ in $G'$, while not increasing the score. For the other direction, note that any newly added production in $G'$ by construction can be replaced  by productions in $G_2$ with the same score.

Observe that the resulting grammar $G'$ is indeed in CNF. Thus, the above steps prove Lemma~\ref{lem:almostcnf}.

\subsection{From LED and RNA-folding to Scored Parsing}
\label{sec:LED}

We show that LED can be reduced to scored parsing on BD grammars. Recall that in LED we are given a context-free grammar $G$ in CNF and a string $\sigma$ of terminals, and we want to compute the smallest edit distance of $\sigma$ to a string $\sigma'$ in the language generated by $G$. The possible edit operations are insertions, deletions, and substitions, and all have cost 1. However, our construction also works if we only allow insertions and deletions (and no substitution).

Recall from the introduction that RNA-folding can be cast as an LED problem without substitutions, where the grammar is given by the productions $S\rightarrow SS \mid \eps$ and $S \to \sigma S \sigma' \mid \sigma' S \sigma$ for any symbol $\sigma \in \Sigma$ with matching symbol $\sigma' \in \Sigma'$. Then if $d$ is the edit distance (using only insertions and deletions) of a given string $\sigma$ to this RNA grammar, then $(|\sigma| - d)/2$ is the maximum number of bases that can be paired in the corresponding RNA sequence. Therefore, RNA folding is covered by our construction for LED without substitutions.

We assume that we are given a scored grammar $G = (N,T,P,S)$ in CNF. In the following we describe how to adapt this grammar. In this procedure, whenever we add a new production $X \to \alpha$ with score $s$, if there already exists the production $X \to \alpha$ with score $s'$, then we only keep the production with lowest cost $\min\{s,s'\}$. Initially, all productions in $G$ get score 0.

\paragraph{Modeling Substitutions}
(For the LED problem without substitutions, simply ignore this paragraph.)
To model substitutions, for any production of the form $X \to c$ (for a non-terminal $X$ and a terminal $c$) in the original grammar, and for each terminal $c' \in T$, we add a production $X \to c'$ with score 1. Note that this allows us to substitute any terminal at a cost of 1 in any derivation $X \to^* \sigma$. In other words, in the resulting scored grammar $\hat G$ the score of any string of terminals $\sigma$ is the minimal number of substitutions to transform $\sigma$ into a string $\sigma'$ in the language generated by $G$.
Note that $\hat G$ is still in CNF. This transformation increases the number of productions by at most $|N| \cdot |T|$.

\paragraph{Modeling Insertions}
Without loss of generality we can assume that for each terminal $a \in T$, there exists a non-terminal $X_a$ and the production $X_a \to a$ with score 0, and this is the only production with $X_a$ on the left-hand-side (if not, we introduce a new non-terminal $X_a$ and the corresponding production, this does not change the generated language or the fact that the grammar is in CNF).
In order to model insertions, we create a new non-terminal $I$, and add the following productions:
\begin{align*}
%& I_a \to a\, (\text{score}=0), \;\; \text{ for every $a \in T$}  \\
& I \rightarrow X_a I \,(\text{score}=1) \mid I X_a\, (\text{score}=1) \mid \varepsilon\, (\text{score}=0), \;\; \text{ for every $a \in T$}  \\
& X \rightarrow X I\, (\text{score}=0) \mid I X\,(\text{score}=0), \;\; \text{ for every nonterminal $X \in N$.}
\end{align*}
Observe that $I$ can generate any string of terminals, and the associated score is the length of the string. Moreover, $I$ can be inserted at any point of a string generated by any non-terminal. 

\paragraph{Modeling Insertions}
In order to model deletions, for any non-terminal $X$, if there exists a production of the form $X \to c$ with terminal $c$, then we add the production
$$X \rightarrow \varepsilon\, (\text{score}=1).$$
This production allows us, in any derivation where $X$ produces a single terminal, to delete this terminal at a cost of 1.

This creates an augmented grammar $G'$ of size polynomial in $|G|$. 
It has been shown in \cite{ap72} that the LED problem on grammar $G$ is equivalent to the scored parsing problem on $G'$. Since $G'$ is almost-CNF by construction, we can use Lemma~\ref{lem:almostcnf} to obtain an equivalent scored grammar $G''$ in CNF, again with size polynomial in $|G|$. In order to use Corollary~\ref{cor:SP} for solving the scored parsing problem on $G''$, it only remains to show that $G'$ (and thus also $G''$) is a BD grammar.

\begin{claim}\label{claim:1}
$G'$ is a $1$-BD grammar.
\end{claim}
\begin{proof}
 Consider any non-terminal $X \in N$ and string of terminals $\sigma$, and let $s := s_{G'}(X,\sigma)$. Then for any terminal $x$, $X\rightarrow IX \rightarrow X_x I X \rightarrow X_x X \rightarrow x X \rightarrow^* x \sigma$ is a valid derivation of $x\sigma$ with score $s+1$. For $X = I$ we similarly use the derivation $I \to X_x I \to x I \to^* x \sigma$. 

For the other direction, consider a derivation $X \to^* x \sigma$ with total score $s'$. In this derivation, the first terminal $x$ must be generated using a production of the form $Y \to x$. By replacing this production with $Y \rightarrow \eps$ we obtain a derivation of the string $\sigma$, while increasing the score by at most 1. 
In total, we obtain $| s_{G'}(X,\sigma) - s_{G'}(X,x\sigma)| \le 1$. The other condition $| s_{G'}(X,\sigma) - s_{G'}(X,\sigma x)| \le 1$ can be shown symmetrically.
\end{proof}

\begin{proposition}
LED and RNA-folding can be reduced to scored parsing problems of $1$-BD grammars. The blow-up in the grammar size is polynomial, and the input string is not changed by the reduction.
\end{proposition}

\subsection{From Optimal Stack Generation to Scored Parsing}\label{sec:osg}

We show that OSG can be reduced to a scored parsing problem on a $3$-BD grammar in almost\nobreakdash-CNF. Recall that in OSG we are given a string $\sigma$ over an alphabet $\Sigma$, and we want to print $\sigma$ by a minimum length sequence of three stack operations: $\textit{push}()$, $\textit{emit}$ (i.e., print the top character in the  stack), and $\textit{pop}$, ending with an empty stack. 

We model this problem as a scored parsing problem as follows.
We have a start symbol $S$ representing that the stack is empty, and a non-terminal $X_c$ for any $c \in \Sigma$ representing that the topmost symbol on the stack is $c$. Moreover, we use a symbol $N_c$ for emitting symbol $c$, and call a production producing $N_c$ a ``pre-emit''. Note that this grammar is already almost-CNF.
\begin{alignat*}{5}
  & S \;&&\to\; \varepsilon \quad && \text{(score 0)} \quad \text{end of string} && \\
  & S \;&&\to\; X_c S \quad && \text{(score 1)} \quad \text{push $c$} && \quad \text{for any $c \in \Sigma$} \\
  & X_c \;&&\to\; N_c X_c \quad && \text{(score 0)} \quad \text{pre-emit $c$} && \quad \text{for any $c \in \Sigma$} \\
  & X_c \;&&\to\; X_{c'} X_c \quad && \text{(score 1)} \quad \text{push $c'$} && \quad \text{for any $c,c' \in \Sigma$} \\
  & X_c \;&&\to\; \varepsilon \quad && \text{(score 1)} \quad \text{pop $c$},&& \quad \text{for any $c \in \Sigma$} \\
  & N_c \;&&\to\; c \quad && \text{(score 1)} \quad \text{emit $c$} && \quad \text{for any $c \in \Sigma$} 
\end{alignat*}
Indeed, these productions model that from an empty stack the only possible operation is to push some symbol $c$, while if the topmost symbol is $c$ then we may (pre-)emit $c$, or push another symbol~$c'$, or pop $c$. It is immediate that the scored parsing problem on this grammar is equivalent to OSG.

For an example, consider the string $bccab$. This string can be generated as follows, where we always resolve the leftmost non-terminal. Note that the suffix of non-terminals always corresponds to the current content of the stack.
\begin{align*}
  S &\to X_b S \to N_b X_b S \to b X_b S \to b X_c X_b S \to b  N_c X_c X_b S \to b c X_c X_b S \to b c N_c X_ c X_b S \\ & \to b c c X_c X_b S \to  b c c X_b S \to b c c X_a X_b S \to b c c N_a X_a X_ b S \to b c c a X_a X_b S \to b c c a X_b S \\ &\to b c c a N_b X_b S \to b c c a b X_b S \to b c c a b S \to b c c a b
\end{align*}

\paragraph{Bounded Differences}
In order to obtain a BD grammar, we slightly change the above grammar by adding the following productions:
\begin{alignat*}{5}
  & N_c \;&&\to\; X_{c'} \quad && \text{(score 1)} \quad \text{helper for BD} && \quad\quad \text{for any $c,c' \in \Sigma$.} 
\end{alignat*}
This does not change the scored language generated by the grammar. Indeed, whenever $N_c$ appears in a derivation starting from $S$, then it was produced by an application of the rule $X_c \to N_c X_c$. Using the new production $N_c \to X_{c'}$ thus results in $X_c \to N_c X_c \to X_{c'} X_c$, with total score~1. However, this derivation can be performed directly using the productions modeling push operations, with the same score. As this is the only way to use the newly added productions in any derivation starting from $S$, the generated language of the grammar is not changed (in fact, only the scored language generated by $N_c$ is changed).

Call the resulting grammar $G$. Note that $G$ is still almost-CNF. We show that it is also BD.

\begin{claim}
$G$ is a $5$-BD grammar.
\end{claim}
\begin{proof}
 Consider a string $\sigma \ne \eps$ over $\Sigma$ and a symbol $x \in \Sigma$. We have to show that for any non-terminal $X$ of $G$, 
 $$\big| s(X,\sigma)-s(X,\sigma x)\big| \leq 5 ~~~~\text{ and }~~~~ \big| s(X,\sigma)-s(X,x\sigma)\big| \leq 5.$$
 Consider a derivation $X \to^* x \sigma$. At some point we produce the first terminal $x$, via the production $N_x \to x$. We change the derivation by instead using $N_x \to X_x \to \eps$, obtaining a derivation of $\sigma$. This increases the score by 1 (as the scores of $N_x \to x$, $N_x \to X_x$, and $X_x \to \eps$ are all 1). Hence, $s(X,\sigma) \le s(X,x\sigma) + 1$. The inequality $s(X,\sigma) \le s(X,\sigma x) + 1$ can be shown symmetrically.
 
 For the other direction, first consider a non-terminal $X$ in $\{S\} \cup \{X_c \mid c \in \Sigma\}$. Consider a derivation $X \to^* \sigma$. Then the adapted derivation $X \to X_x X \to N_x X_x X \to x X_x X \to x X \to^* x \sigma$ increases the score by 3 and generates $x \sigma$. 
 
 Similarly, to generate $\sigma x$, note that $X$ is always the rightmost symbol during the whole derivation, until we delete it with the rule $X \to \eps$. At the point in the derivation $X \to^* \sigma$ where we delete $X$ via $X \to \eps$, instead use the derivation $X \to X_x X \to N_x X_x X \to x X_x X \to x X \to x$, to produce $x$ at a cost of 3. Then the adapted derivation generates $\sigma x$.
 
 For non-terminals $X = N_c$ (for any $c \in \Sigma$) we argue as follows. If the first step of the derivation is $N_c \to X_{c'}$, then we can instead argue about $X_{c'}$, which we have done above. Otherwise, the derivation is $N_c \to c$, and $\sigma = c$. Then to produce $x c$ we instead use the derivation 
 $$ N_c \to X_c \to X_x X_c \to N_x X_x X_c \to x X_x X_c \to x X_c \to x N_c X_c \to x c X_c \to x c, $$
 at a cost of 6, increasing the score of $N_c \to c$ by 5. The case $\sigma x$ is symmetric. 
 
 In all cases, the scores of $\sigma$ and $x \sigma$ (or $\sigma x$) differ by at most 5.
\end{proof}

Together with Lemma~\ref{lem:almostcnf} we now obtain the following.

\begin{proposition}
OSG can be reduced to a scored parsing problem of a BD grammar. The size of the grammar is polynomial in $|\Sigma|$, and the input string is not changed by the reduction.
\end{proposition}

\paragraph{Acknowledgments} The authors would like to thank Uri Zwick for numerous discussions during the initial stages of this project.

\bibliographystyle{plain}
\bibliography{../clustershort}

%\bibliographystyle{plain}
%\bibliography{cluster}
\end{document}